\documentclass[hidelinks,onefignum,onetabnum]{siamart251104}

\usepackage{amsfonts}
\usepackage{booktabs}
\usepackage{graphicx}
\usepackage{tabularx}

\newsiamremark{remark}{Remark}
\newsiamremark{hypothesis}{Hypothesis}
\crefname{hypothesis}{Hypothesis}{Hypotheses}
\newsiamthm{claim}{Claim}
\newsiamremark{fact}{Fact}
\crefname{fact}{Fact}{Facts}

\headers{Posets Corresponding to the Partition Problem}{S. Kubo}

\title{Partially Ordered Sets Corresponding to the Partition Problem\thanks{An extended abstract version of this paper appeared in Proceedings of the 25th International Symposium on Fundamentals of Computation Theory, 2025, \url{https://doi.org/10.1007/978-3-032-04700-7_23}.
\funding{This work was partially supported by MEXT Leading Initiative for Excellent Young Researchers Grant Number JPMXS0320200347 and JSPS KAKENHI Grant Number JP26K17039.}}}

\author{Susumu Kubo\thanks{Faculty of Informatics, Showa Women's University, Tokyo, Japan
  (\email{s-kubo@swu.ac.jp}).}
}

\usepackage{amsopn}

\ifpdf
\hypersetup{
  pdftitle={Partially Ordered Sets Corresponding to the Partition Problem},
  pdfauthor={S. Kubo}
}
\fi

\usepackage{bm}
\usepackage{enumitem}
\newsiamremark{example}{Example}
\crefname{example}{Example}{Examples}
\usepackage{amsmath, amssymb}
\usepackage{tikz}
\usetikzlibrary{calc}

\newcommand{\mb}[1]{\mathbf{#1}}
\newcommand{\pl}{\leqslant}

\begin{document}

\maketitle

\begin{abstract}
The partition problem is a well-known NP-complete problem. 
We focus on its optimization version.  
We propose two partially ordered sets (posets) corresponding to the partition problem and develop an order-theoretic framework for solving it. The first poset is order-isomorphic to a well-known poset whose structure is related to solutions of the subset sum problem, while the second is a subposet of the first and plays a crucial role in this paper. The partial order characterizes the dominance relations between subsets that hold uniformly across all instances.
We first show several properties of the two posets, such as size, height, and width (the size of the largest antichain, i.e., the largest set of pairwise incomparable elements). The two posets have the same width, which is $\Theta(2^n / n^{3/2})$ for $n$ congruent to $0$ or $3$ modulo $4$; this exponential width indicates the hardness of the partition problem.  
We then prove that the initial candidate solutions are the elements of the second poset, whose size is $2^{n} - 2 \binom{n}{\lfloor n/2 \rfloor}$. Since a partition corresponds to two elements of the poset, the number of initial candidate partitions is half of that, i.e., $2^{n-1} - \binom{n}{\lfloor n/2 \rfloor}$.
We prove that the candidate solutions can be further reduced based on the partial order, and we establish a necessary and sufficient condition, phrased in terms of the second poset, for a subset to attain the optimal value.
Building on this optimality criterion, we finally derive several polynomially solvable cases from the structure of the second poset. 
Our approach offers a useful tool for structural analysis of the partition problem.
\end{abstract}



\section{Introduction}\label{sec1}
The partition problem is formulated as follows:
given positive integers $c_1, \ldots, c_n$, the objective is to decide whether there exists a subset $S$ of $\{1, \ldots, n\}=:[n]$ such that $\Delta(S) = 0$, where 
\begin{equation}\label{eq.delta}
\Delta(S):=
\sum_{i \in S}c_i - \sum_{i \in [n] \setminus S}c_i. 
\end{equation}
In this paper, we focus on its optimization version: 
to find a subset $S$ of $[n]$ that minimizes the absolute value $|\Delta(S)|$. 
Without loss of generality, we index the integers in nonincreasing order ($c_1 \geq c_2 \geq \cdots \geq c_n$). 
Although the $c_i$'s are originally positive, in what follows we only require $c_n \geq 0$; allowing zero entries loses no generality and is convenient because some proofs use instances with zero entries.

The partition problem is one of Karp's 21 NP-complete problems \cite{karp1972}
and is also one of Garey and Johnson's six basic NP-complete problems \cite{Garey1979}.
Therefore, the optimization version is known to be NP-hard.
It has important applications in task scheduling \cite{Coffman1991,Tsai1992}. In fact, it is equivalent to makespan minimization on two identical parallel machines.
 
The problem has been investigated from various perspectives. 
Several optimal algorithms give exact solutions in time exponential in $n$ \cite{Garey1979,Horowitz1974,Korf1998,Schroeppel1981}.
On the pseudopolynomial side, the classical dynamic programming approach has been improved to near-linear time in the total input sum \cite{Bringmann2017,Koiliaris2019}.
Moreover, a variety of heuristic and metaheuristic algorithms have been developed \cite{Alidaee2005,Fuksz2013,Johnson1991,Karmarkar1982,Ruml1996}. 
An ``easy-hard'' phase transition was observed and the solution structure has been studied from the point of view of statistical mechanics \cite{borgs2001,Gent1996,Hartmann2008,Hayes2002,Mertens1998,Mertens2006}. 
An algebraic expression was presented based on max-plus algebra \cite{kubo2020}. 

Furthermore, the number of solutions to the subset sum problem, a special case of which is the partition problem, has been studied. 
The subset sum problem with positive inputs is to decide whether there exists a subset $S$ of $[n]$ such that
\[
\sum_{i \in S} c_i = t,
\]
where $t$ is a given constant (the target). 
Let $\nu(\{c_1, \ldots, c_n\},t)$ denote the number of solutions. 
It is known that, if $c_1, \ldots, c_n$ are distinct positive integers, then
\[
\nu \left( \{c_1, \ldots, c_n\},t \right) \leq 
\nu \left([n], \left\lfloor \frac{n(n+1)}{4} \right\rfloor \right)
\]
\cite{Lindstrom1970,Stanley1991}. 
In fact, $\nu \left([n], \lfloor n(n+1)/4 \rfloor \right)$ is sequence A025591 in the On-Line Encyclopedia of Integer Sequences \cite{OEIS} and is $\Theta(2^n / n^{3/2})$ for $n$ congruent to $0$ or $3$ modulo $4$ \cite{Sullivan2013}.
The inequality is established using the properties of a partially ordered set (poset).

The poset is a well-known poset, denoted by $M(n)$ in Stanley's paper \cite[Subsection 4.1.2]{Stanley1991}. 
 The elements of $M(n)$ are all the subsets of $[n]$. Let $A$ and $B$ be two subsets of $[n]$ with elements $a_1, a_2, \ldots, a_j$ and $b_1, b_2, \ldots, b_k$, respectively, and assume that the elements of $A$ and $B$ are ordered in decreasing order. The partial order is defined as follows: $A \pl B$ if and only if $j \leq k$ and $a_i \leq b_i$ for $i \in [j]$. 

In this paper we develop a poset-based approach to the optimization version of the partition problem; to the best of our knowledge, posets have not previously been applied to the problem in this way. 
\subsection{Main Results}\label{ssec:main}
To define two posets that correspond to the partition problem, 
we first consider the set $\mathbb{R}^n$ of $n$-
dimensional vectors of real numbers. 
Let $\mb{x}=(x_1, \ldots, x_n)\in \mathbb{R}^n$ and $\mb{y}=(y_1, \ldots, y_n)\in \mathbb{R}^n$, 
and define a partial order on the set:
$\mb{x} \preceq \mb{y}$
if and only if
$x_1 + \cdots + x_k \leq y_1 + \cdots + y_k$ for all $k\in[n]$.
In other words, $\mb{x} \preceq \mb{y}$ means that any entry of the cumulative sum of $\mb{x}$ is less than or equal to the corresponding one of $\mb{y}$. This comparison of cumulative sums is reminiscent of the order relations studied in the theory of majorization \cite{Marshall2011}.

Let us define a poset $P(n)$ that corresponds to the partition problem. 
Let $P(n)$ be the set of all $n$-dimensional vectors with entries in $\{1, -1\}$. The size of $P(n)$ is of course $2^n$. 
The partial order is defined to be the restriction of $\preceq$ to $P(n)$, denoted by $\preceq_P$. The poset is order-isomorphic to $M(n)$ under the mapping $f$ from $P(n)$ to $M(n)$ defined by the following rule: 
Given $\mb{v}=(v_1, \dots, v_n) \in P(n)$, we define $f(\mb{v})$ by $i \in f(\mb{v})$ if and only if $v_{n+1-i
}=1$. 

\begin{example}
The Hasse diagrams of $M(5)$ and $P(5)$ are shown in \Cref{fig_M5,fig_P5}, respectively. 
We can see that the two posets are order-isomorphic under $f$. 
For example, $f((\phantom{-}1,-1,\phantom{-}1,-1,-1))= \{5, 3\}$. 
Note that a partial order is shown horizontally due to space limitations, though it is usually shown vertically.
The horizontal position indicates rank (see \Cref{ssec.preli}).
\end{example}

\begin{remark}
An element of $P(n)$ can also be regarded as a lattice path on the integers that starts at $0$ and at each step moves $+1$ or $-1$.
The relation $\mb{v} \preceq_P \mb{w}$ means that the path corresponding to $\mb{v}$ does not ``exceed'' the one corresponding to $\mb{w}$. \Cref{fig_path} shows an example.

\begin{figure}[ht]
  \centering
    \includegraphics[scale = 0.76]{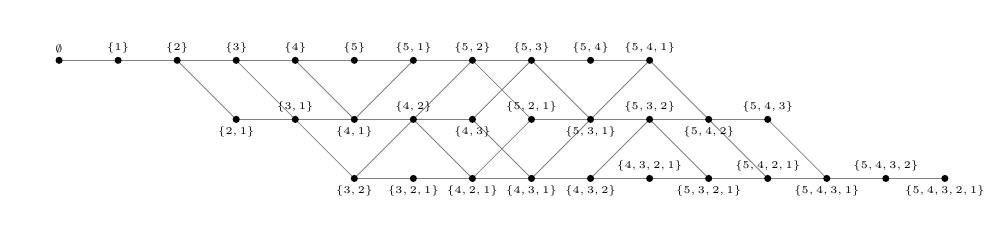}
    \vspace{-3mm}
  \caption{The Hasse diagram of $M(5)$.}
  \label{fig_M5}
\end{figure}
\begin{figure}[H]
  \centering
  \includegraphics[scale = 0.76]{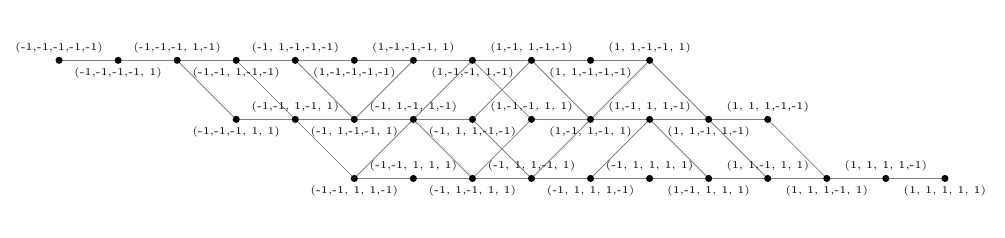}
   \vspace{-3mm}
  \caption{The Hasse diagram of $P(5)$.}
  \label{fig_P5}
\end{figure}

\begin{figure}[ht]
  \centering
    \includegraphics[height=4cm]{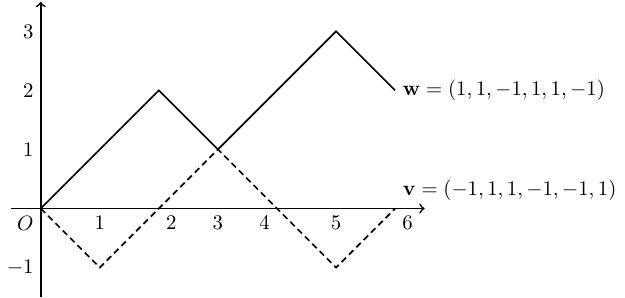}
  \caption{The relation $\mb{v}\preceq_P \mb{w}$ in $P(6)$ is shown as lattice paths.}
  \label{fig_path}
\end{figure}
\end{remark}

Let us define another poset $Q(n)$ to be the subset of $P(n)$ obtained by deleting from $P(n)$ the elements that are less than or greater than the zero vector, that is,
$
Q(n) = P(n) \setminus \{ {\mb{v} \in P(n)} \colon {\mb{v}\preceq \mb{0}} \allowbreak\text{ or } {\mb{0} \preceq \mb{v}}\}, 
$
where $\mb{0}=(0, \dots, 0)$. 
The partial order is defined to be the restriction of $\preceq_P$ to $Q(n)$, denoted by $\preceq_Q$.
The poset $Q(n)$ is obtained by truncating both the upper and lower parts of $P(n)$. 
The size of $Q(n)$ is proved to be $2^{n} - 2 \binom{n}{\lfloor n/2 \rfloor}$. 

\begin{example}
Let us consider $Q(n)$ for $n \leq 6$. We can easily see that $Q(1)=Q(2)=\emptyset$ and $Q(3)=\{(1,-1,-1), (-1,1,1)\}$, the elements of which are incomparable. 
The Hasse diagrams of $Q(4)$, $Q(5)$, and $Q(6)$ are shown in \Cref{fig_Q45}.

\begin{figure}[ht]
\renewcommand{\figurename}{}
\renewcommand{\thefigure}{}
  \begin{minipage}[b]{0.3\linewidth}
    \centering
    \includegraphics[scale = 1]{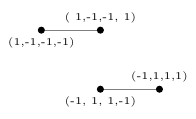}
    \centering
    \caption{\hspace{-2mm}\textcolor{white}{$\blacksquare$}$Q(4)$}
  \end{minipage}
  \hfill
  \begin{minipage}[b]{0.65\linewidth}
    \centering
    \includegraphics[scale =1]{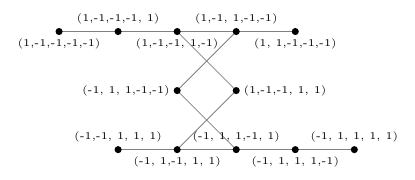}
    \caption{\hspace{-2mm}\textcolor{white}{$\blacksquare$}$Q(5)$}
  \end{minipage}
  \begin{minipage}[b]{\linewidth}
      \centering
    \includegraphics[scale =.95]{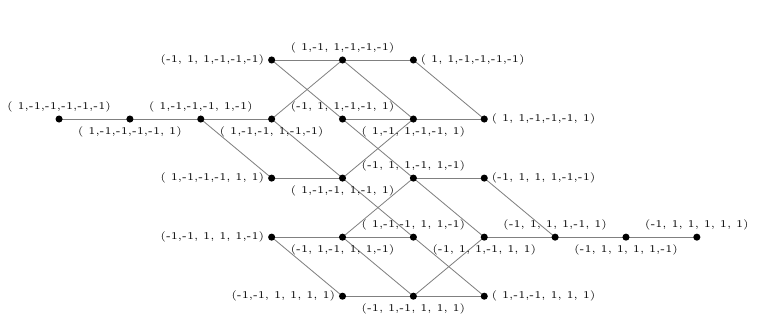}
    \caption{\hspace{-2mm}\textcolor{white}{$\blacksquare$}$Q(6)$}
  \end{minipage}
\renewcommand{\figurename}{Fig.}
\renewcommand{\thefigure}{\arabic{figure}}
\addtocounter{figure}{-3} 
  \caption{The Hasse diagrams of $Q(4)$, $Q(5)$, and $Q(6)$.}
  \label{fig_Q45}
\end{figure}

\end{example}

To investigate the relationship between the partition problem and the posets $P(n)$ and $Q(n)$, 
we rewrite the difference (\ref{eq.delta}) of the problem.
There is a one-to-one correspondence between $P(n)$ and the set of all subsets of $[n]$.
We introduce a mapping $\mb{p}$ from $2^{[n]}$ to $P(n)$. 
To each subset $S$ of $[n]$, we assign $\mb{p}(S)=(p_1, \dots, p_n)$ defined by 
\[
p_i =
\begin{cases}
  \phantom{-}1 & \text{ if }\quad i \in S,\\
  -1 & \text{ if }\quad i \notin S.
\end{cases}
\]
Let $\mb{c}= (c_1, c_2, \ldots, c_n)$. 
Then
\[
\Delta(S)=\mb{p}(S)\cdot \mb{c},
\]
where the dot indicates the usual inner product.
Note that $\mb{p}([n]\setminus S) = -\mb{p}(S)$ and hence $\Delta([n]\setminus S)= - \mb{p}(S)\cdot \mb{c}$, the absolute value of which is equal to  $\lvert \Delta(S) \rvert$. 
A pair of two elements $\mb{p}(S)$ and $-\mb{p}(S)$ corresponds to the partition $\{S, \;[n]\setminus S\}$.

The partial order now acquires a concrete meaning in optimization: $\mb{p}(S)\cdot \mb{c} \leq \mb{p}(S')\cdot \mb{c}$ holds for every instance $\mb{c}$ if and only if $\mb{p}(S) \preceq_P \mb{p}(S')$ (\Cref{prop:d&order}). In contrast to the instance-dependent dominance relations commonly used in enumerative algorithms, the posets $P(n)$ and $Q(n)$ thus encode exactly the instance-independent dominance structure of the partition problem.

We show that it suffices to consider $Q(n)$ for solving the partition problem.
Let us call a subset $S$ of $[n]$ \textit{dispensable} if there exists a subset $S'$ of $[n]$, not equal to $S$ or $[n]\setminus S$, such that $\lvert \Delta(S')\rvert \leq \lvert \Delta(S)\rvert$ for every instance $\mb{c}$; an optimal solution can then always be found without examining $S$. The \textit{initial candidate solutions} are defined as the subsets that are not dispensable. This set is determined independently of any algorithm.
A \textit{candidate solution} is a subset that remains under consideration at a given stage of the optimization process.

\begin{theorem}[Characterization of Initial Candidate Solutions]\label{th:sol}
    The initial candidate solutions to the partition problem are precisely the subsets corresponding to the elements of $Q(n)$. 
\end{theorem}
The number of initial candidate partitions is half of the size of $Q(n)$, i.e., $2^{n-1} - \binom{n}{\lfloor n/2 \rfloor}$ (sequence A093387 in \cite{OEIS}), since a partition comprises a subset and its complement.
Furthermore, candidate solutions can be reduced based on the partial order of $Q(n)$. 

\begin{theorem}[Order-Based Pruning of Candidates]\label{th:signQ}
  Given an instance $\mb{c}$ of the partition problem, the following holds for every $\mb{v} \in Q(n)$:
 \begin{enumerate}
    \item If $\mb{v}\cdot \mb{c} \geq 0$,
  then $\lvert \mb{v}\cdot \mb{c}\rvert \leq \lvert \mb{w}\cdot \mb{c}\rvert$ for every $\mb{w} \in Q(n)$ that satisfies $\mb{w} \preceq_Q -\mb{v}$ or $\mb{v} \preceq_Q \mb{w}$.
  \item If $\mb{v}\cdot \mb{c} \leq 0$,
then $\lvert \mb{v}\cdot \mb{c}\rvert \leq \lvert \mb{w}\cdot \mb{c}\rvert$ for every $\mb{w} \in Q(n)$ that satisfies $\mb{w} \preceq_Q \mb{v}$ or $-\mb{v} \preceq_Q \mb{w}$.
    \end{enumerate}
\end{theorem}

Building on this pruning, we establish a necessary and sufficient condition for an element of $Q(n)$ to attain the optimal value. To state it, we call an element $\mb{v} \in Q(n)$ \emph{regular} with respect to $\mb{c}$ if $\mb{u}\cdot \mb{c} < \mb{v}\cdot \mb{c}$ for every $\mb{u} \in Q(n)$ with $\mb{u} \prec_Q \mb{v}$. Since $\mb{u} \preceq_Q \mb{v}$ implies $\mb{u}\cdot \mb{c} \leq \mb{v}\cdot \mb{c}$ (\Cref{prop:d&order}), regularity means that no element below $\mb{v}$ shares its inner product with $\mb{c}$.

\begin{theorem}[Optimality Criterion]\label{th:opt}
Given an instance $\mb{c}$ of the partition problem, for $\mb{v} \in Q(n)$ the value $\mb{v}\cdot \mb{c}$ is the optimal value if and only if $\mb{v}\cdot \mb{c} = 0$, or there exists a regular element $\mb{w} \in Q(n)$ with $\mb{w} \preceq_Q \mb{v}$ and $\mb{w}\cdot \mb{c} = \mb{v}\cdot \mb{c}$ such that:
\begin{romannum}
    \item $\mb{w} \cdot \mb{c} > 0$;
    \item $(\mb{w} + \mb{w}_1)\cdot \mb{c} \leq 0$ for every element $\mb{w}_1$ of $Q(n)$ covered by $\mb{w}$, if any;
    \item $(\mb{w} + \mb{w}_2)\cdot \mb{c} \leq 0$ or  $(\mb{w} - \mb{w}_2)\cdot \mb{c} \leq 0$ for every element $\mb{w}_2$ of $Q(n)$ incomparable to both $\mb{w}$ and $-\mb{w}$, if any.
\end{romannum}
\end{theorem}

Since $\mb{w}$ covers at most $n$ elements, the size of the criterion in \Cref{th:opt} is governed by the number of elements incomparable to $\pm\mb{w}$; when this number is small, optimality can be decided in polynomial time. The simplest cases are the minimal and maximal elements themselves.

Let $\ell = \lfloor (n-1)/2 \rfloor$. For $k \in \{0, 1, \dots, \ell \}$ we denote by $\mb{m}_k$ the element of $Q(n)$
\begin{equation*}\label{eq_max}
(\underbrace{1, \ldots, 1}_{k}, \underbrace{-1, \ldots, -1}_{k+1}, 1, \ldots, 1).
\end{equation*}
Thus, $-\mb{m}_k$ stands for the element of $Q(n)$
\begin{equation*}\label{eq_min}
(\underbrace{-1, \ldots, -1}_{k}, \underbrace{1, \ldots, 1}_{k+1}, -1, \ldots, -1).
\end{equation*}
In fact, $-\mb{m}_k$ and $\mb{m}_k$ are the minimal and maximal elements, respectively. 

\begin{corollary}[Polynomially Solvable Case for Minimal Elements]\label{th:min}
Let $-\mb{m}_k$ be a minimal element of $Q(n)$, where $k \in \{0, 1, \dots, \ell \}$. Given an instance $\mb{c}$ of the partition problem, the value $-\mb{m}_k \cdot \mb{c}$ is the optimal value if and only if $-\mb{m}_k \cdot \mb{c} \geq 0$; in that case the subset $S$ with $\mb{p}(S)=-\mb{m}_k$ is an optimal solution.
\end{corollary}

\Cref{th:min} is a generalization of the obvious fact that the subset $\{1\}$ is an optimal solution (with optimal value $c_1 - (c_2 + \cdots + c_n)$ when $c_1 - (c_2 + \cdots + c_n) \geq 0$), corresponding to the minimal element $-\mb{m}_0 = (1, -1, \dots, -1)$.

Using the addition and swap operators $\mathcal{A}^{(k)}$ and $\mathcal{S}^{(j,k)}$ defined in \Cref{sec2}, which map an element to a larger one,
we immediately obtain the following corollary for the maximal elements. 

\begin{corollary}[Polynomially Solvable Case for Maximal Elements]\label{cor:maximal}
Let $\mb{m}_k$ be a maximal element of $Q(n)$, where $k \in \{0, 1, \dots, \ell \}$. Given an instance $\mb{c}$ of the partition problem, if the following inequalities hold, then $\mb{m}_k \cdot \mb{c}$ is the optimal value and the subset $S$ with $\mb{p}(S)=\mb{m}_k$ is an optimal solution:
\begin{romannum}
\item $\mb{m}_k \cdot \mb{c} \geq 0$;
\item $\left(\mb{m}_k -\mathcal{S}^{(k,k+1)}(-\mb{m}_k)\right)\cdot \mb{c} \leq 0$ (applicable for  $k\neq 0$);
\item $\left(\mb{m}_k-\mathcal{A}^{(n)}(-\mb{m}_k)\right)\cdot \mb{c} \leq 0$ (applicable for $k\neq (n-1)/2$).
\end{romannum}
\end{corollary}

Unlike \Cref{th:min}, only sufficiency is asserted here: a maximal element need not be regular, and its full necessary and sufficient condition is the general one of \Cref{th:opt}. Conditions (ii) and (iii) of \Cref{cor:maximal}, which compare $\mb{m}_k$ with elements obtained by the swap and addition operators, would be hard even to formulate without the poset $Q(n)$.

Moreover, the same criterion yields a polynomial-time optimality test for an explicitly described family of elements near the minimal and maximal ones, for which it reduces to a constant number of inequalities (\Cref{sec3}, \Cref{cor:extended}).

Our approach provides a new tool for the structural analysis of the partition problem. Two formulation choices essential to it---encoding subsets by $\{1,-1\}$-vectors and keeping the sign-opposite pairs $\{\mb{v},-\mb{v}\}$ rather than halving the candidates---are discussed in \Cref{sec4}.

\subsection{Organization of the Paper}
The paper is organized as follows. 
In \Cref{sec2} we review basic poset concepts and the properties of $P(n)$ and $Q(n)$. 
In \Cref{sec3} we explore the relationship between the partition problem and these posets; we prove our three main theorems and derive the polynomially solvable cases (\Cref{cor:extended}).
Finally, \Cref{sec4} concludes the paper with remarks and open problems.

\section{Two Posets Corresponding to the Partition Problem}\label{sec2}
\subsection{Preliminaries}\label{ssec.preli}
In this subsection, we review some standard facts on posets and give some properties of the poset $M(n)$. 

Let $(P, \preceq)$ be a finite poset. We simply denote the poset by $P$ and define the relation $\prec$ by $x \prec y$ if and only if $x \preceq y$ and $x \neq y$. 
An element $y$ {\it covers} $x$ if $x \prec y$ and there is no $z \in P$ such that $x \prec z \prec y$. 
A subset of $P$ is a {\it chain} if it is totally ordered, and
a subset of $P$ is an {\it antichain} if no pair of elements in it is comparable.
A subset $F$ of $P$ is an {\it up-set} ({\it down-set}) if $x \in F$ and $x \preceq y$ ($y \preceq x$, resp.) imply $y \in F$; the {\it up-set} ({\it down-set}) {\it of an element} is the set of elements greater (less, resp.) than or equal to it.
The {\it height} ({\it width}) of $P$ is the largest size of a chain (antichain, resp.) in $P$.
An element of $P$ is a {\it maximal} ({\it minimal}) {\it element} if it is not less (greater, resp.) than any other element. 
An element of $P$ is the {\it greatest} ({\it least}) {\it element} if it is greater (less, resp.) than any other element.

A poset $P$ is {\it graded}\footnote{There are several competing definitions for graded posets.} if it is equipped with a {\it rank function} $\rho$ from $P$  to $\mathbb{Z}$ that satisfies:
\begin{enumerate}
  \item if $x \prec y$, then $\rho(x)< \rho(y)$;
  \item if $y$ covers $x$, then $\rho(y) = \rho(x) + 1$.
\end{enumerate}
An element $x \in P$ has {\it rank} $i$ if $\rho(x)=i$. By setting the minimum value of the rank function to zero, any graded poset $P$ can be partitioned into $P = \bigcup_{i=0}^{r} P_i$, where $P_i = \{x \in P \colon \rho(x)=i \}$. We denote the size of $P_i$
by $s_i$. 
The poset $P$ is {\it rank-symmetric} if $s_i = s_{r-i}$ for each $i$.
The poset $P$ is {\it rank-unimodal} if $s_0 \leq s_1 \leq \cdots \leq s_i \geq s_{i+1} \geq \cdots \geq s_r$ for some $i$, and $P$ is {\it Sperner} if the width of $P$ is equal to $\max_i \{s_i\}$. 
The width is lower bounded by $\max_i \{s_i\}$ 
since every $P_i$ is an antichain.
In a Sperner poset, the largest one among $P_0, P_1, \ldots, P_r$ provides an antichain of maximum size.

A poset $P$ is a \textit{lattice} if any two elements $x, y \in P$  have a least upper bound and a greatest lower bound, denoted by $x \vee y$ and $x \wedge y$ respectively. A lattice $L$ is \textit{distributive} if $x \vee (y \wedge z) = (x \vee y) \wedge (x \vee z)$ for each $x, y, z \in L$.

The poset $M(n)$ has the following properties \cite{Proctor1982,Stanley1980}:
\begin{itemize}
  \item The rank function with minimum rank 0 is the sum of elements of a subset.
  The maximum rank is $n(n+1)/2$.
  \item The poset $M(n)$ is rank-symmetric, rank-unimodal, and Sperner.
  Therefore, its width is equal to the size of the middle rank, which is $\nu \left([n], \left\lfloor n(n+1)/4 \right\rfloor \right)$. 
  \item Its height is $n(n+1)/2 + 1$.
  \item The poset $M(n)$ is a distributive lattice with 
  the greatest element $[n]$ and the least element $\emptyset$.
\end{itemize}

\subsection{Properties of the Two Posets}
We now establish the basic properties of $P(n)$ and $Q(n)$.
We first need to become familiar with some elementary properties of the partial order $\preceq$. Recall the partial order defined in \Cref{ssec:main}. Let $\mb{x}=(x_1, \ldots, x_n)\in \mathbb{R}^n$ and $\mb{y}=(y_1, \ldots, y_n)\in \mathbb{R}^n$. 
The relation $\mb{x} \preceq \mb{y}$ holds
if and only if
\begin{align*}
  x_1 &\leq y_1,\\
  x_1 + x_2 &\leq y_1 + y_2,\\
  &\;\; \vdots \\
  x_1 + x_2 + \cdots + x_n &\leq y_1 + y_2 + \cdots + y_n.
\end{align*}
We record three elementary consequences.
\begin{proposition}\label{prop:real}
  For $\mb{x}, \mb{y}, \mb{z} \in \mathbb{R}^n$ and $\alpha \in \mathbb{R}$, the following statements hold: 
  \begin{enumerate}
    \item $\mb{x} \preceq \mb{y} \Longleftrightarrow \mb{x}+\mb{z} \preceq \mb{y}+\mb{z}$.
    \item $\mb{x} \preceq \mb{y}$ and $\,\alpha \geq 0 \Longrightarrow \alpha\mb{x} \preceq \alpha\mb{y}$.
    \item $\mb{x} \preceq \mb{y}$ and $\,\alpha \leq 0 \Longrightarrow \alpha\mb{y} \preceq \alpha\mb{x}$.
  \end{enumerate}
\end{proposition}
\begin{proof}
  Straightforward. 
\end{proof}
Next, we investigate the poset $P(n)$. As defined in \Cref{ssec:main}, the ground set of $P(n)$ consists of all $n$-dimensional vectors with entries in $\{1, -1\}$ 
and the order relation is the restriction of $\preceq$ to the ground set, denoted by $\preceq_P$.
In order to examine the poset $P(n)$, we define two operators. 
Let us denote the $i$th entry of a vector $\mb{v}$ by $v_i$ hereafter. 

\begin{definition}[Addition operator]
  For $k \in [n]$, the {\it addition operator} $\mathcal{A}^{(k)}$ takes the input $\mb{v} \in P(n)$ with $v_k=-1$ and sets the $k$th entry to $1$, 
  that is,
  \[
  \left(\mathcal{A}^{(k)}\mb{v}\right)_i =
  \begin{cases}
   v_k+2=1 & \text{ if }\; i=k\\
    v_i & \text{ if }\; i \neq k.
  \end{cases}
  \]
  For other elements $\mb{v}$, the output $\mathcal{A}^{(k)}\mb{v}$ is not defined.
\end{definition}


\begin{definition}[Swap operator]
  For $j$, $k \in [n]$ with $ j < k$, the {\it swap operator} $\mathcal{S}^{(j, k)}$ takes the input $\mb{v}\in P(n)$ with $v_{j}=-1$ and $v_k = 1$, and 
  swaps $v_j$ and $v_k$, that is,
  \[
  \left(\mathcal{S}^{(j,k)}\mb{v}\right)_i =
  \begin{cases}
    v_k=1 
    & \text{ if } \; i=j\\
    v_j=-1 
    & \text{ if } \; i=k\\
    v_i & \text{ if } \; i \notin \{j, k\}.
  \end{cases}
  \]
  For other elements $\mb{v}$, the output $\mathcal{S}^{(j,k)}\mb{v}$ is not defined.
\end{definition}

%

We can show some relationships between the operators and the partial order $\preceq_P$. 
\begin{proposition}\label{prop:<}
For $\mb{v}, \mb{w} \in P(n)$, the relation $\mb{v} \prec_P \mb{w}$ holds if and only if $\mb{w}$ is obtained from $\mb{v}$ by applying addition and/or swap operators.
\end{proposition}

\begin{proof}
Assume that $\mb{w}$ is obtained from $\mb{v}$ by applying addition and/or swap operators.
For each $\mb{u} \in P(n)$ and any possible integers $j$ and $k$, we have $\mb{u} \prec_P \mathcal{A}^{(k)}\mb{u}$ and $\mb{u} \prec_P \mathcal{S}^{(j,k)}\mb{u}$. Therefore, $\mb{v} \prec_P \mb{w}$.

Conversely, assume that $\mb{v} \prec_P \mb{w}$. 
Let $\mb{x} = \mb{w} - \mb{v}$. Then $x_i$ is $\pm 2$ or $0$ for $i \in [n]$. 
Denote the indices of $x_i$'s equal to $-2$ by $k_1, \ldots, k_l$ in increasing order. 
Since any entry of the cumulative sum of $\mb{x}$ is nonnegative, we can pair each entry equal to $-2$ with an entry equal to 2 on the left, that is, we can choose distinct indices $j_1, \ldots, j_l$ of $x_i$'s equal to $2$ such that $j_1 < k_1, \ldots, j_l < k_l$.
Hence, $\mb{x}$ is decomposed as follows:

\begin{multline*}
\mb{x}
= (\underbrace{0, \ldots, 0}_{j_1-1}, 2, \underbrace{0, \ldots, 0}_{k_1-j_1-1}, -2, 0, \ldots, 0)
   + \dots + (\underbrace{0, \ldots, 0}_{j_l-1}, 2, \underbrace{0, \ldots, 0}_{k_l-j_l-1}, -2, 0, \ldots, 0)\\
   + (\underbrace{0, \ldots, 0}_{i_1-1}, 2, 0, \ldots, 0)
   + \dots + (\underbrace{0, \ldots, 0}_{i_m-1}, 2, 0, \ldots, 0),
\end{multline*}
  where $i_1, \dots, i_m$ are the indices of the remaining $2$'s. 
  Since the indices are all distinct, we have 
  $
  \mb{w}=\mathcal{A}^{(i_m)}\cdots \mathcal{A}^{(i_1)}\mathcal{S}^{(j_{l}, k_{l})}\cdots \mathcal{S}^{(j_1, k_1)}\mb{v}.
  $
This means that $\mb{w}$ is obtained from $\mb{v}$ by applying addition and/or swap operators.
\end{proof}

\begin{proposition}\label{prop:cover}
For $\mb{v}, \mb{w} \in P(n)$, the element $\mb{w}$ covers $\mb{v}$ if and only if $\mb{w}$ is obtained from $\mb{v}$ by applying exactly one of the following operators: $\mathcal{A}^{(n)}$ or one of  $\mathcal{S}^{(1,2)}, \mathcal{S}^{(2,3)}, \ldots, \mathcal{S}^{(n-1,n)}$. 
\end{proposition}

\begin{proof}
Assume that $\mb{w}$ covers $\mb{v}$, that is, $\mb{v} \prec_P \mb{w}$ and there is no $\mb{u}\in P(n)$ such that $\mb{v} \prec_P \mb{u} \prec_P \mb{w}$. Suppose, contrary to our claim, that $\mb{w}$ is not obtained from $\mb{v}$ by one of the operators.
We have the following three cases:

\noindent \textbf{Case 1:} $\mb{w}$ is obtained from $\mb{v}$ by two or more addition and/or swap operators.  Then there exists $\mb{u}\in P$ such that $\mb{v} \prec_P \mb{u} \prec_P \mb{w}$ since one of the operators always makes the input larger. This is a contradiction.

\noindent \textbf{Case 2:} $\mb{w}= \mathcal{A}^{(i)}\mb{v}$ for some $i\in [n-1]$. Then $v_i=-1$. 
If $v_{i+1}= 1$ then $\mb{w} =\mathcal{A}^{(i+1)} \mathcal{S}^{(i, i+1)} \mb{v}$.
If $v_{i+1}=-1$ then $\mb{w} = \mathcal{S}^{(i, i+1)} \mathcal{A}^{(i+1)} \mb{v}$.
Both lead to Case 1, a contradiction.

\noindent \textbf{Case 3:} $\mb{w} = \mathcal{S}^{(j,k)}\mb{v}$ for some $j$ and $k$ with $k-j\geq 2$. Then $v_j=-1$ and $v_k=1$. If $v_{j+1}=1$ then $\mb{w} = \mathcal{S}^{(j+1, k)} \mathcal{S}^{(j, j+1)} \mb{v}$.
If $v_{j+1}=-1$ then $\mb{w} = \mathcal{S}^{(j, j+1)} \mathcal{S}^{(j+1, k)}\mb{v}$.
Both lead to Case 1, a contradiction.

Conversely, assume that $\mb{w}$ is obtained from $\mb{v}$ by applying one of the operators. 
It is clear that $\mb{v} \prec_P \mb{w}$. 
Suppose, contrary to our claim, that there exists $\mb{u} \in P(n)$ such that $\mb{v} \prec_P \mb{u} \prec_P \mb{w}$. If $\mb{w} = \mathcal{A}^{(n)}\mb{v}$, then $v_i = w_i$ for $i \in[n-1]$, $v_n =-1$ and $w_n =1$. Thus, $u_i = v_i = w_i$ for $i \in[n-1]$ and $u_n \in \{1, -1\}$, so $\mb{u}$ equals $\mb{v}$ or $\mb{w}$, which is a contradiction.
If $\mb{w} = \mathcal{S}^{(k,k+1)}\mb{v}$ for some $k \in [n-1]$, then $v_i = w_i$ for $i \notin \{k, k+1\}$, $v_k =-1$, $v_{k+1} =1$, $w_k =1$ and $w_{k+1} =-1$. Thus, $u_i = v_i =w_i$ for $i \in [k-1]$, 
$u_k + u_{k+1}=0$, and $u_i = v_i =w_i$ for $i \in [n]\setminus [k+1]$, 
which is also a contradiction. 
\end{proof}

\begin{proposition}\label{prop:iso}
The poset $P(n)$ is order-isomorphic to $M(n)$.
\end{proposition}
\begin{proof}
We need to show that $\mb{v} \preceq_P \mb{w}$ if and only if $f(\mb{v}) \pl f(\mb{w})$ for each $\mb{v}$ and $\mb{w}$ in $P(n)$.

Assume that $\mb{v} \preceq_P \mb{w}$. 
From \Cref{prop:<} it follows that $\mb{w}$ is obtained from $\mb{v}$ by applying addition and/or swap operators. 
An addition operator increases the number of $1$'s of the input by one and a swap operator shifts one of the $1$'s of the input to the left.
Now let $\mb{v}$ have $m$ $1$'s.
Then $\mb{w}$ has at least $m$ $1$'s, which implies that $f(\mb{w})$ has at least as many elements as $f(\mb{v})$. 
The $m$ left-most $1$'s of $\mb{w}$ do not lie more to the right than the $1$'s of $\mb{v}$, which implies that in decreasing order the $i$th element of $f(\mb{w})$ is not less than that of $f(\mb{v})$ for $i \in [m]$. Hence, $f(\mb{v}) \pl f(\mb{w})$.

Conversely, assume that $f(\mb{v}) \pl f(\mb{w})$. Let $f(\mb{v})=\{a_1, \ldots, a_j\}$ and $f(\mb{w}) = \{b_1, \ldots, b_k\}$ in decreasing order. Then $j \leq k$, which means that $\mb{w}$ has at least as many $1$'s as $\mb{v}$,
and $a_i \leq b_i$ for $i \in [j]$, which means that the $j$ left-most $1$'s of $\mb{w}$ do not lie more to the right than the $1$'s of $\mb{v}$. 
By the properties of addition and swap operators, 
$\mb{w}$ is obtained from $\mb{v}$ by applying addition and/or swap operators, that is, $\mb{v} \preceq_P \mb{w}$. 
\end{proof}

The poset $P(n)$ has the following symmetry. 
\begin{proposition}\label{prop:sym}
For $\mb{v}, \mb{w} \in P(n)$, the relation $\mb{v} \prec_P \mb{w}$ holds if and only if $-\mb{w} \prec_P -\mb{v}$.
\end{proposition}
\begin{proof}
For each $\mb{v} \in P(n)$, it holds that $-\mb{v} \in P(n)$. The statement follows immediately from \Cref{prop:real}(3). 
\end{proof}

It follows from \Cref{prop:iso} and its proof and \Cref{prop:sym} that the poset $P(n)$ has the following properties:
\begin{itemize}
  \item The rank function with minimum rank 0 is 
  \begin{equation}\label{eq:rankf}
  \frac{1}{2}\left(\mb{v}\cdot \mb{c}_0 + \frac{n(n+1)}{2}\right)
  \end{equation}
  for $\mb{v} \in P(n)$, where $\mb{c}_0 = (n, n-1, \ldots, 2, 1)$. 
  \item The poset $P(n)$ is rank-symmetric, rank-unimodal, and Sperner. 
  \item Its width is $\nu \left([n], \left\lfloor n(n+1)/4 \right\rfloor \right)$. 
  \item Its height is $n(n+1)/2 + 1$. 
  \item The poset $P(n)$ has a certain symmetry: $\mb{v} \prec_P \mb{w} \Longleftrightarrow -\mb{w} \prec_P -\mb{v}$.
  \item The poset $P(n)$ is a distributive lattice.
  The greatest element and the least one are respectively $(1, \ldots, 1)$ and $(-1, \ldots, -1)$.
\end{itemize}

We investigate the other poset $Q(n)$. 
Hereafter, let $n \geq 3$. We show some properties of $Q(n)$.

Let $R_{+}(n) = \{ \mb{v} \in P(n) \colon \mb{0} \prec \mb{v}\}$ and 
$R_{-}(n) = \{ \mb{v} \in P(n) \colon \mb{v}\prec \mb{0} \}$.
We can see that the subposets $R_+(n)$ and $R_-(n)$ are sublattices of $P(n)$. 
The least element of $R_+(n)$ is $(1, -1, 1, -1, \ldots)$, the rank of which in $P(n)$ is
$n(n+2)/4$ if $n$ is even and 
$(n+1)^2/4$ if $n$ is odd. 
The elements of $R_+(n)$ are not in the middle rank of $P(n)$. This is true for $R_-(n)$ by the symmetry of $P(n)$. Therefore, we obtain the next proposition.

\begin{proposition}
The width of $Q(n)$ is the same as that of $P(n)$.
\end{proposition}

For the size of $Q(n)$, the next proposition holds.
\begin{proposition}\label{prop:Qsize}
The size of $Q(n)$ is $2^n - 2\binom{n}{\lfloor n/2 \rfloor}$.
\end{proposition}
\begin{proof}
By definition, $\# Q(n) = \# P(n)- \# R_+(n)-\# R_-(n)$, where $\#$ indicates the size of a set. 
Since $\# P(n)=2^n$ and $\# R_+(n)= \# R_-(n)$, it is sufficient to show that $\# R_+(n)= \binom{n}{\lfloor n/2 \rfloor}$.
Let us regard the elements of $R_+(n)$ as lattice paths.
The size is the number of $n$-step paths from 0 never going below 0.
Let $X_n$ be this count, with $X_0 = 1$. Then we have
\[
X_n =
\begin{cases}
    2 X_{n-1}- C_{(n-1)/2} 
    & \text{ if $n$ is odd}\\
    2 X_{n-1} & \text{ if $n$ is even},
\end{cases}
\]
where $C_m$ is the $m$th Catalan number.
Indeed, a path of length $n-1$ can be extended by a step $+1$ or $-1$ unless it ends at height $0$, in which case only the step $+1$ is available; paths of length $n-1$ from $0$ to $0$ never going below $0$ exist only when $n-1$ is even, and their number is $C_{(n-1)/2}$.
Solving the recurrence (or by induction on $n$), we obtain $X_n = \binom{n}{\lfloor n/2 \rfloor}$, which is sequence A001405 in \cite{OEIS}. 
\end{proof}

The poset $Q(n)$ has a similar symmetry. 
\begin{proposition}\label{prop:symQ}
For $\mb{v}, \mb{w} \in Q(n)$, the relation $\mb{v} \prec_Q \mb{w}$ holds if and only if $-\mb{w} \prec_Q -\mb{v}$.
\end{proposition}
\begin{proof}
For each $\mb{v} \in Q(n)$, we show that $-\mb{v} \in Q(n)$.
Suppose, contrary to our claim, that $-\mb{v} \notin Q(n)$. 
Then $-\mb{v} \prec \mb{0}$ or $\mb{0} \prec -\mb{v}$, which contradicts $\mb{v} \in Q(n)$. 
The statement follows immediately from \Cref{prop:real}(3). 
\end{proof}

\begin{proposition}\label{prop:Qgraded}
The poset $Q(n)$ is graded and one of the rank functions is the same as the rank function (\ref{eq:rankf}) of $P(n)$.
\end{proposition}
\begin{proof}
    Denote the rank function of $P(n)$ by $\rho$. 
    We show that a rank function of $Q(n)$ is $\rho$.
    Since $Q(n)$ is a subposet of $P(n)$, if $\mb{v} \prec_Q \mb{w}$ then $\rho(\mb{v}) < \rho(\mb{w})$. 
    When $\mb{w}$ covers $\mb{v}$ in $Q(n)$, suppose, contrary to our claim, that $\rho(\mb{w}) \geq \rho(\mb{v})+2$. 
    Then there exists some $\mb{u} \in P(n) \setminus Q(n)$ such that $\mb{v} \prec_P \mb{u} \prec_P \mb{w}$ and $\rho(\mb{u})=\rho(\mb{v})+1$. Since $\mb{u}$ is not in $Q(n)$, $\mb{u} \prec \mb{0}$ or $\mb{0} \prec \mb{u}$. 
    If $\mb{u} \prec \mb{0}$ then $\mb{v} \prec \mb{0}$.
    If $\mb{0} \prec \mb{u}$ then $\mb{0} \prec \mb{w}$.
    Either case is a contradiction. 
\end{proof}

\begin{corollary}\label{cor:coverQ}
For $\mb{v}, \mb{w} \in Q(n)$, the element $\mb{w}$ covers $\mb{v}$ in $Q(n)$ if and only if $\mb{w}$ covers $\mb{v}$ in $P(n)$; that is, the cover relations of $Q(n)$ are exactly the cover relations of $P(n)$ whose two elements both lie in $Q(n)$.
\end{corollary}
\begin{proof}
If $\mb{w}$ covers $\mb{v}$ in $P(n)$, then $\mb{v} \prec_Q \mb{w}$ because both elements lie in $Q(n)$ by hypothesis, and no element of $Q(n) \subseteq P(n)$ lies strictly between them; hence $\mb{w}$ covers $\mb{v}$ in $Q(n)$. Conversely, if $\mb{w}$ covers $\mb{v}$ in $Q(n)$, then $\rho(\mb{w}) = \rho(\mb{v}) + 1$ by \Cref{prop:Qgraded}, so no element of $P(n)$ lies between $\mb{v}$ and $\mb{w}$, that is, $\mb{w}$ covers $\mb{v}$ in $P(n)$.
\end{proof}

\begin{proposition}\label{prop:maximal}
  The maximal (minimal) elements of $Q(n)$ are
  $\mb{m}_0$, $\mb{m}_1$, $\ldots, \mb{m}_\ell\:$ 
  $(-\mb{m}_0$, $-\mb{m}_1$, $\ldots, -\mb{m}_\ell$, resp.).
\end{proposition}

\begin{proof}
We first prove that $\mb{m}_k$ is maximal for $k \in \{0, 1, \dots, \ell\}$.
It is sufficient to show that the elements covering $\mb{m}_k$ on $P(n)$ are not in $Q(n)$.
If $n$ is even, then the last entry of $\mb{m}_k$ is always $1$ for each $k$. From this fact and \Cref{prop:cover} it follows that the unique element covering $\mb{m}_k$ is $\mathcal{S}^{(2k+1, 2k+2)} \mb{m}_k$, that is, 
\begin{equation*}
(\underbrace{1, \ldots, 1}_{k}, \underbrace{-1, \ldots, -1}_{k}, 1, -1, 1, \ldots, 1),
\end{equation*}
which is greater than $\mb{0}$.
If $n$ is odd, then the last entry of $\mb{m}_k$ might be $-1$. The unique element covering $\mb{m}_k$ is $\mathcal{S}^{(2k+1, 2k+2)} \mb{m}_k$ 
if $k < \ell$, and $\mathcal{A}^{(n)} \mb{m}_k$, that is,
\begin{equation*}
(\underbrace{1, \ldots, 1}_{k}, \underbrace{-1, \ldots, -1}_{k}, 1)
\end{equation*}
if $k = \ell$.
Both are greater than $\mb{0}$.

We next prove that no other elements are maximal. Consider a maximal element $\mb{v} \in Q(n)$. 
Now we distinguish two cases.

\noindent \textbf{Case 1:} $\mb{v}$ contains no $(-1, 1)$ sequence. Then $\mb{v}$ is a vector with $1$'s up to some position (possibly none) and then $-1$'s for the rest of the entries, that is, $(1, \ldots, 1)$, $(-1, \ldots, -1)$ or $(1,\ldots, 1, -1,\ldots, -1)$. 
The first two are obviously not in $Q(n)$.
The element $\mathcal{A}^{(n)}\mb{v}$ covers $\mb{v}$ in $P(n)$, so it is not in $Q(n)$ by the maximality of $\mb{v}$; moreover, $\mathcal{A}^{(n)}\mb{v} \prec \mb{0}$ would imply $\mb{v} \prec \mb{0}$, contradicting $\mb{v} \in Q(n)$. Hence $\mathcal{A}^{(n)}\mb{v} \in R_{+}(n)$, that is, all cumulative sums of $\mathcal{A}^{(n)}\mb{v}$ are nonnegative. Since the operator $\mathcal{A}^{(n)}$ changes only the last cumulative sum, raising it by $2$, the other cumulative sums of $\mb{v}$ coincide with those of $\mathcal{A}^{(n)}\mb{v}$ and are nonnegative, so the only possibly negative cumulative sum of $\mb{v}$ is the last one; as $v_1 + \cdots + v_{n-1} \geq 0$ and $v_n = -1$,
the sum $v_1 + \cdots + v_{n}$ must be equal to $-1$.
Therefore, $\mb{v}$ is 
\begin{equation*}
(\underbrace{1, \ldots, 1}_{(n-1)/2}, \underbrace{-1, \ldots, -1}_{(n-1)/2},-1),
\end{equation*}
where $n$ is odd. 
This is exactly $\mb{m}_{\ell}$.

\noindent \textbf{Case 2:} $\mb{v}$ contains $(-1, 1)$ sequences. 
We assume that the left-most sequence lies at the $j$th and $(j+1)$th positions. Then $\mathcal{S}^{(j,j+1)}\mb{v} \in R_+(n)$: this element covers $\mb{v}$ in $P(n)$ (\Cref{prop:cover}), so it is not in $Q(n)$ by the maximality of $\mb{v}$, and $\mathcal{S}^{(j,j+1)}\mb{v} \prec \mb{0}$ would imply $\mb{v} \prec \mb{0}$.
The swap raises the $j$th cumulative sum by $2$ and leaves the others unchanged, which therefore coincide with those of $\mathcal{S}^{(j,j+1)}\mb{v} \in R_+(n)$ and are nonnegative, so the only possibly negative cumulative sum of $\mb{v}$ is the $j$th one; as $v_1 + \cdots + v_{j-1} \geq 0$ and $v_j = -1$,
the sum $v_1 + \cdots + v_j$ must be equal to $-1$.
If $\mb{v}$ contains another $(-1, 1)$ sequence at the $j'$th and $(j'+1)$th positions, then $\mathcal{S}^{(j',j'+1)}\mb{v} \in Q(n)$, which contradicts the maximality. 
Therefore, the subsequences $(v_1, \ldots, v_{j-1})$ and $(v_{j+2}, \ldots, v_n)$ contain no $(-1, 1)$ sequence. 
By a similar argument to Case 1, they are $(1, \ldots, 1)$, $(-1, \ldots, -1)$, $(1,\ldots, 1, -1,\ldots, -1)$ or empty. 
Since $v_1 + \cdots + v_{j-1}=0$, we get 
\begin{equation*}
(v_1, \ldots, v_{j-1})= (\underbrace{1, \ldots, 1}_{(j-1)/2}, \underbrace{-1, \ldots, -1}_{(j-1)/2}),
\end{equation*}
where $j$ is odd (if $j=1$ then the subsequence is empty). 
On the other hand, if $v_n =-1$ then $\mathcal{A}^{(n)}\mb{v} \in Q(n)$, which contradicts the maximality. Hence, $(v_{j+2}, \ldots, v_n)$ is $(1, \dots, 1)$ or empty. 
Consequently, $\mb{v}$ is
\begin{equation*}
(\underbrace{1, \ldots, 1}_{(j-1)/2}, \underbrace{-1, \ldots, -1}_{(j-1)/2}, -1, 1, \underbrace{1,\ldots, 1}_{n-(j+1)}),
\end{equation*}
which is exactly $\mb{m}_{k}$, where $k=(j-1)/2$ and $j \leq n-1$.
Note that $k \in \{0, 1, \dots, \ell\}$ if $n$ is even and $k \in \{0, \dots, \ell-1\}$ if $n$ is odd. 

By symmetry, if $\mb{v}$ is a maximal element, then $-\mb{v}$ is a minimal element.  
\end{proof}



%


\begin{proposition}\label{prop:Qheight}
 The height of $Q(n)$ is  
  \begin{align*}
    &\frac{n(n-1)}{2}-\left(n-\frac{3}{2}\ell\right)(\ell+1)+1 &&\text{ for }\: n \leq 7, \text{ and}&&&\\
    &\frac{(n-2)(n-3)}{2}+1 &&\text{ for }\: n\geq 8. 
 \end{align*}
 \end{proposition}

We first show the next lemma.
\begin{lemma}\label{lem:chain}
 In the poset $Q(n)$ with $n \geq 4$, a minimal element $-\mb{m}_k$ is less than a maximal element $\mb{m}_{k'}$ with $k' \neq k$, that is, $-\mb{m}_k \prec_Q \mb{m}_{k'}$.
\end{lemma}
\begin{proof}
Let $k_0 = \min\{k, k'\}$ and $k_1 =\max\{k, k'\}$. 
If $k_0 < k_1 \leq 2k_0 + 1$, then 
\begin{equation*}
\mb{m}_{k'}-(-\mb{m}_k) = (\underbrace{2, \ldots, 2}_{k_0}, 
\underbrace{0, \ldots, 0}_{k_1 - k_0}, 
\underbrace{-2, \ldots, -2}_{2k_0 + 1 - k_1},
\underbrace{0, \ldots, 0}_{2(k_1 - k_0)}, 2, \ldots, 2),
\end{equation*}
which is greater than $\mb{0}$ since $k_0 \geq 2k_0 + 1 - k_1$.
If $2k_0 + 1 < k_1$, then
\begin{equation*}
\mb{m}_{k'}-(-\mb{m}_k) = (\underbrace{2, \ldots, 2}_{k_0}, 
\underbrace{0, \ldots, 0}_{k_0+1}, 
\underbrace{2, \ldots, 2}_{k_1 - (2k_0 + 1)},
\underbrace{0, \ldots, 0}_{k_1 + 1}, 2, \ldots, 2),
\end{equation*}
which is evidently greater than $\mb{0}$. 
\end{proof}
Note that in the case where $n=3$, the maximal elements $\mb{m}_0=(-1,1,1)$ and $\mb{m}_1=(1,-1,-1)$ are also minimal elements. 

\begin{proof}[Proof of \Cref{prop:Qheight}]
We first consider the case where $n\geq 4$. 
    From \Cref{lem:chain} it follows that two elements $-\mb{m}_k$ and $\mb{m}_{k'}$ with $k' \neq k$ are comparable. 
    Let us take a maximal chain\footnote{A chain is a maximal chain if no other element can be added to it without losing the property of being totally ordered.} containing the two elements. 
    Since $Q(n)$ is graded, the size of the maximal chain is $\rho(\mb{m}_{k'}) -\rho(-\mb{m}_{k}) + 1$, where $\rho$ denotes the rank function (\ref{eq:rankf}) of $P(n)$ and $Q(n)$.
    Therefore, the height is
    \begin{equation}\label{eq_height}
    \max_{k \neq k'} \left\{\rho(\mb{m}_{k'}) -\rho(-\mb{m}_{k}) + 1\right\}.
    \end{equation}
    Since $\rho(\mb{v}) =\left(\mb{v} \cdot \mb{c}_0 + n(n+1)/2 \right)/2$ for $\mb{v} \in Q(n)$, we have
    \begin{align*}
    \rho(\mb{m}_{k'}) &= \frac{1}{2}n(n+1)-n(k'+1)+\frac{3}{2}k'(k'+1),\\
    \rho(-\mb{m}_{k}) &= n(k+1)-\frac{3}{2}k(k+1).
    \end{align*}
    Substituting these expressions into (\ref{eq_height}), we obtain
    \begin{equation*}\label{eq_height'}
    \max_{k \neq k'} \left\{\frac{3}{2}\left(k - \frac{2n-3}{6}\right)^2 
    + \frac{3}{2}\left(k' - \frac{2n-3}{6}\right)^2
    - 3 \left(\frac{2n-3}{6}\right)^2 + \frac{n(n-3)}{2} + 1\right\}.
    \end{equation*}
The objective function is symmetric in $k$ and $k'$, and we can assume that $k > k'$. 
In the $k$-$k'$ plane the farthest point from the center point $\left((2n-3)/6, (2n-3)/6\right)$ provides the maximum value. 
Therefore, the candidate points are $(1, 0)$ and $(\ell, 0)$. 
Determining the position of the center point relative to the line $k=(1+\ell)/2$, equidistant from the two points, we conclude that the optimal point is $(\ell, 0)$ for $n \leq 7$ and $(1,0)$ for $n \geq 8$. 
Hence, the height is $n(n-1)/2-(n-3\ell/2)(\ell+1)+1$ for $n \leq 7$ and $(n-2)(n-3)/2+1$ for $n\geq 8$.

In the case where $n=3$, the height is $1$, which is obtained by the expression above for $n\leq 7$.   
\end{proof}


Next, we establish a property of the minimal elements, which will play a crucial role in the next section.
\begin{proposition}\label{lem:compara}
For any $\mb{v} \in Q(n)$ and any minimal element of $Q(n)$,  
the minimal element is comparable with either $\mb{v}$ or $-\mb{v}$;  
more precisely, either $-\mb{m}_k \preceq \mb{v}$ or $-\mb{m}_k \preceq -\mb{v}$ holds for each $k \in \{0, 1, \dots, \ell\}$.
\end{proposition}

\begin{proof}
Fix $k$ and let $U = \{\mb{x} \in P(n) \colon -\mb{m}_k \preceq \mb{x}\}$. Now $-\mb{m}_k$ is
  \begin{equation*}
  (\underbrace{-1, \ldots, -1}_{k}, \underbrace{1, \ldots, 1}_{k+1}, \underbrace{-1, \ldots, -1}_{n-2k-1}).
\end{equation*}
We write $\mb{x} \in P(n)$ as $(\mb{x}_1, \mb{x}_2)$ with $\mb{x}_1 \in P(2k+1)$ and $\mb{x}_2 \in P(n-2k-1)$, and claim that
\[
U = \{ (\mb{x}_1, \mb{x}_2) \in P(n) \colon \mb{x}_1 \text{ contains at most } k \text{ entries equal to } -1 \}.
\]
If $\mb{x} \in U$, then the $(2k+1)$th cumulative sum of $\mb{x}$ is at least that of $-\mb{m}_k$, which is $1$; hence $\mb{x}_1$ contains at most $k$ entries equal to $-1$.
Conversely, let $\mb{x}_1$ contain exactly $m$ entries equal to $-1$, where $m \leq k$. Then $\mb{x}_1$ is greater than or equal to
\begin{equation*}
(\underbrace{-1, \ldots, -1}_{k}, \underbrace{1, \ldots, 1}_{k+1})\in P(2k+1)
\end{equation*}
in $P(2k+1)$: the $i$th cumulative sum of the latter is $\max\{-i,\, i-2k\}$, while that of $\mb{x}_1$ is at least $\max\{-i,\, i-2m\} \geq \max\{-i,\, i-2k\}$.
Moreover, any $\mb{x}_2 \in P(n-2k-1)$ is greater than or equal to the least element $(-1, \ldots, -1)$ of $P(n-2k-1)$, and the $(2k+1)$th cumulative sum of $\mb{x}$ is $2k+1-2m$, which is at least $1$. Combining these facts entrywise in the cumulative sums yields $-\mb{m}_k \preceq \mb{x}$, which proves the claim. Consequently,
\[
\# U = \left( \binom{2k+1}{0} + \cdots +\binom{2k+1}{k} \right) 2^{\,n-2k-1} = 2^{2k} \cdot 2^{\,n-2k-1}= 2^{n-1}.
\]
The set $U$ contains no sign-opposite pair $\{\mb{x}, -\mb{x}\}$: if $-\mb{m}_k \preceq \mb{x}$ and $-\mb{m}_k \preceq -\mb{x}$, then the latter gives $\mb{x} \preceq \mb{m}_k$ (\Cref{prop:real}), whence $-\mb{m}_k \preceq \mb{m}_k$; since $-\mb{v} \preceq \mb{v}$ is equivalent to $\mb{0} \preceq \mb{v}$, this contradicts $\mb{m}_k \in Q(n)$.
Now $P(n)$ consists of exactly $2^{n-1}$ sign-opposite pairs, so $U$, having $2^{n-1}$ elements and meeting each pair at most once, meets each pair exactly once.
In particular, for each element $\mb{v}$ of $Q(n)$, either $-\mb{m}_k \preceq_P \mb{v}$ or $-\mb{m}_k \preceq_P -\mb{v}$ holds.  
\end{proof}

We conclude this section by proving that $Q(n)$ is rank-unimodal. At first sight this seems delicate: the rank sizes of $Q(n)$ arise from those of $P(n)$ by subtracting the contributions of $R_+(n)$ and $R_-(n)$, and a difference of unimodal sequences need not be unimodal. The proof combines the structure established above with a classical algebraic property of $M(n)$.

Let $P$ be a finite graded poset with rank decomposition $P = \bigcup_{i=0}^{r} P_i$, and let $\mathbb{C}P_i$ denote the complex vector space with basis $P_i$. A linear map $U$ on $\bigoplus_{i} \mathbb{C}P_i$ is \emph{order-raising} if for every $x \in P_i$ the vector $Ux$ lies in the span of the elements covering $x$. By \cite{Proctor1982,Stanley1980}, the poset $M(n)$---and hence $P(n)$ by \Cref{prop:iso}---admits an order-raising operator $U$ whose iterate $U^{r-2i} \colon \mathbb{C}P_i \to \mathbb{C}P_{r-i}$ is bijective for every $i \leq r/2$; we use only the consequence that $U \colon \mathbb{C}P_{i} \to \mathbb{C}P_{i+1}$ is injective for every $i < r/2$, which holds because the bijection $U^{r-2i}$ factors through it.

\begin{theorem}\label{th:unimodal}
The poset $Q(n)$ is rank-unimodal.
\end{theorem}

\begin{proof}
Write $Q(n) = P(n) \setminus (R_+(n) \cup R_-(n))$, put $q_i = \#(Q(n) \cap P_i)$, and let $r = n(n+1)/2$. Let $U$ be an order-raising operator for $P(n)$ as above.

The set $R_+(n)$ is an up-set of $P(n)$: if $\mb{0} \prec \mb{v}$ and $\mb{v} \prec \mb{w}$, then $\mb{0} \prec \mb{w}$. Its minimum rank is attained at its least element $(1, -1, 1, -1, \ldots)$ and, as computed earlier in this section, equals $n(n+2)/4$ for even $n$ and $(n+1)^2/4$ for odd $n$. Both exceed $\lceil r/2 \rceil = \lceil n(n+1)/4 \rceil$: the differences are at least $n/4 - 1/2 > 0$ and $(n+1)/4 - 1/2 > 0$, respectively. By the symmetry of $P(n)$ (\Cref{prop:sym}), the set $R_-(n) = -R_+(n)$ is a down-set.

We first show that $(q_i)$ is nondecreasing up to the middle. Let $1 \leq i \leq \lceil r/2 \rceil$; we claim that $q_{i-1} \leq q_i$. Let $\mb{v} \in Q(n) \cap P_{i-1}$ and let $\mb{w}$ cover $\mb{v}$. Then $\mb{w} \notin R_-(n)$, since otherwise the down-set $R_-(n)$ would contain $\mb{v}$; and $\mb{w} \notin R_+(n)$ because the rank of $\mb{w}$ is $i \leq \lceil r/2 \rceil$, below the minimum rank of $R_+(n)$. Hence every element covering an element of $Q(n) \cap P_{i-1}$ lies in $Q(n) \cap P_i$, and since $U$ is order-raising,
\[
U\bigl(\operatorname{span}(Q(n) \cap P_{i-1})\bigr) \subseteq \operatorname{span}(Q(n) \cap P_i).
\]
As $i-1 < r/2$, the map $U \colon \mathbb{C}P_{i-1} \to \mathbb{C}P_i$ is injective, and so is its restriction to the subspace $\operatorname{span}(Q(n) \cap P_{i-1})$. Comparing dimensions yields $q_{i-1} \leq q_i$.

Finally, $Q(n)$ is rank-symmetric (\Cref{prop:symQ} with the shared rank function of \Cref{prop:Qgraded} gives $q_i = q_{r-i}$), so the sequence $(q_i)$ is nonincreasing from the middle on. It is therefore unimodal.
\end{proof}

Consequently, $Q(n)$ shares the rank-symmetry, rank-unimodality, and Sperner property of $P(n)$. We summarize the properties of $Q(n)$ established in this section:
\begin{itemize}
  \item The rank function with minimum rank 0 is 
   \[
   \frac{1}{2}\left(\mb{v}\cdot \mb{c}_0 + \frac{n(n+1)}{2}\right) -n
   \]
  for $\mb{v} \in Q(n)$.
  \item Its size is $2^n - 2\binom{n}{\lfloor n/2 \rfloor}$.
  \item The poset $Q(n)$ is rank-symmetric, rank-unimodal, and Sperner.
  \item Its width is the same as that of $P(n)$, i.e., $\nu \left([n], \left\lfloor n(n+1)/4 \right\rfloor \right)$.
  \item Its height is $n(n-1)/2-\left(n-3\ell/2\right)(\ell+1)+1$ for $n \leq 7$ and $(n-2)(n-3)/2+1$ for $n\geq 8$. 
  \item The poset $Q(n)$ also has the same symmetry as $P(n)$: \\
  $\mb{v} \prec_Q \mb{w} \Longleftrightarrow -\mb{w} \prec_Q -\mb{v}$.
  \item The minimal elements and the maximal ones are respectively $-\mb{m}_0, -\mb{m}_1, \ldots,$ $ -\mb{m}_\ell$ and $\mb{m}_0, \mb{m}_1, \ldots, \mb{m}_\ell$. 
\end{itemize}

\section{Relationship between the Partition Problem and the Posets}\label{sec3}

We now discuss the relationship between the partition problem and the posets $P(n)$ and $Q(n)$.
\subsection{Basic Relationship}
We first relate the difference $\mb{p}(S)\cdot \mb{c}$ in the partition problem with the partial orders $\preceq_P$ and $\preceq_Q$. 

\begin{proposition}\label{prop:d&order}
  Given two subsets $S$ and  $S'$ of $[n]$, 
  the inequality $\mb{p}(S) \cdot \mb{c} \leq \mb{p}(S') \cdot \mb{c}$ holds for any instance $\mb{c}$
  if and only if $\mb{p}(S) \preceq_P \mb{p}(S')$.
\end{proposition}

To show the proposition, we first introduce another expression of the difference $\mb{p}(S)\cdot \mb{c}$. 
Let $\mb{d} = (d_1, \dots, d_n)$ defined by $d_i = c_i-c_{i+1}$ for $i \in [n]$, where $c_{n+1}=0$. Note that $\mb{d} \geq \mb{0}$, which is a componentwise inequality. 
Denote the cumulative sum of $\mb{p}(S)$ by $\mb{r}(S)$, that is, $\mb{r}(S)=(p_1, p_1+p_2, \dots, p_1+\dots+p_n)$ when $\mb{p}(S)=(p_1, p_2, \dots, p_n)$; equivalently, $\mb{r}(S)$ is the sequence of heights of the lattice path corresponding to $\mb{p}(S)$ (see \Cref{fig_path}).
Then $\mb{p}(S)\cdot \mb{c} = \mb{r}(S)\cdot \mb{d}$.
In addition, the relation $\mb{p}(S) \preceq_P \mb{p}(S')$ can be succinctly written as the componentwise inequality $\mb{r}(S) \leq \mb{r}(S')$. 

\begin{proof}[Proof of \Cref{prop:d&order}]
  We denote $\mb{p}(S)$ and $\mb{p}(S')$ by $(p_1, \ldots, p_n)$ and $(p'_1, \ldots,$ $ p'_n)$, respectively.

We first assume that 
$\mb{p}(S) \cdot \mb{c} \leq \mb{p}(S') \cdot \mb{c}$ for any $\mb{c}$. 
Suppose, contrary to our claim, that $p_1+\cdots +p_k > p'_1 + \cdots + p'_k$ for some $k \in [n]$.
Take 
\[
\mb{c}' = (\underbrace{1, \ldots, 1}_{k}, 0, \ldots, 0).
\]
Then $\mb{p}(S) \cdot \mb{c}' > \mb{p}(S') \cdot \mb{c}'$, which is a contradiction.

We next prove the converse. 
We assume that $\mb{p}(S) \preceq_P \mb{p}(S')$, that is, $\mb{r}(S) \leq \mb{r}(S')$. 
Suppose, contrary to our claim, that $\mb{p}(S) \cdot \mb{c} > \mb{p}(S') \cdot \mb{c}$ for some $\mb{c}$. Then we have $(\mb{r}(S) - \mb{r}(S'))\cdot \mb{d} > 0$. Since $\mb{r}(S) - \mb{r}(S')$ is nonpositive and $\mb{d}$ is nonnegative, the left-hand side is nonpositive, a contradiction.  
\end{proof}

Since $\preceq_Q$ is the restriction of $\preceq_P$ to $Q(n)$, the same equivalence holds in $Q(n)$, with $\preceq_Q$ in place of $\preceq_P$.

Next, we identify subsets that can be ignored in solving the partition problem. 
\begin{lemma}\label{lem_order&d}
If a subset $S$ of $[n]$ satisfies $\mb{p}(S) \notin Q(n)$, that is, $\mb{p}(S) \prec \mb{0}$ or $\mb{0} \prec \mb{p}(S)$, then $S$ is dispensable.
\end{lemma}

\begin{proof}
We denote $\mb{p}(S)$ by $(p_1, \ldots, p_n)$.
For the complement of $S$, we have $\mb{p}([n]\setminus S) = - \mb{p}(S)$ and $\lvert \mb{p}([n]\setminus S) \cdot \mb{c} \rvert = \lvert \mb{p}(S) \cdot \mb{c} \rvert$. 
Without loss of generality, we can thus assume that $\mb{p}(S) \prec \mb{0}$. 
We need to consider two cases:

\noindent \textbf{Case 1:} $\mb{p}(S)$ contains no $(-1, 1)$ sequence. Since $\mb{p}(S) \prec \mb{0}$, it holds that $S = \emptyset$. 
Take $\{n\}$ as $S'$. 
Then $\mb{p}(S) \prec_P \mb{p}(S')$ and $\mb{p}(S') \prec \mb{0}$.
It follows from the former relation and \Cref{prop:d&order} that  $\mb{p}(S) \cdot \mb{c} \leq \mb{p}(S') \cdot \mb{c}$ holds for any $\mb{c}$. 
The latter relation is equivalent to $\mb{r}(S') \leq \mb{0}$. The inner product of $\mb{r}(S')$ with the nonnegative vector $\mb{d}$ is nonpositive. Thus, $\mb{r}(S') \cdot \mb{d} = \mb{p}(S') \cdot \mb{c}\leq 0$ for any $\mb{c}$. 
Therefore, $\mb{p}(S) \cdot \mb{c} \leq \mb{p}(S') \cdot \mb{c} \leq 0$ for any $\mb{c}$. 

\noindent \textbf{Case 2:} $\mb{p}(S)$ contains $(-1, 1)$ sequences. 
We assume that the left-most sequence lies at the $j$th and $(j+1)$th positions. 
Take $S'$ such that  $\mb{p}(S') = \mathcal{S}^{(j, j+1)}\mb{p}(S)$. Then $\mb{p}(S)\prec_P \mb{p}(S')$ by \Cref{prop:<}.
On the other hand, subtracting $-\mb{p}(S)$ from $\mb{p}(S')$, we have
\[
\mb{p}(S')- (-\mb{p}(S)) = (2 p_1, \ldots, 2 p_{j-1}, 0, 0, 2 p_{j+2}, \ldots, 2 p_n).
\]
This implies $\mb{p}(S') \prec_P -\mb{p}(S)$, since $\mb{p}(S) \prec \mb{0}$.
Consequently, it follows from \Cref{prop:d&order} that $\mb{p}(S)\cdot \mb{c} \leq \mb{p}(S')\cdot \mb{c} \leq -\mb{p}(S)\cdot \mb{c}$ for any $\mb{c}$. 
\end{proof}

Thus, \Cref{lem_order&d} implies that we need not consider elements of $P(n) \setminus Q(n)$. We then show that, in general, all elements of $Q(n)$ must be considered. 

\begin{lemma}\label{lem_order&d&r}
Let a subset $S$ of $[n]$ satisfy $\mb{p}(S) \in Q(n)$.
Then there exists some instance $\mb{c}$ such that $ \lvert \mb{p}(S) \cdot \mb{c}\rvert < \lvert \mb{p}(S')\cdot \mb{c}\rvert$ for any subset $S'$ of $[n]$ that is not equal to $S$ or $[n]\setminus S$.
\end{lemma}

\begin{proof}
We will prove this by contradiction. 
Suppose that for any $\mb{c}$ there exists some $S'$ such that 
$ \lvert \mb{p}(S') \cdot \mb{c}\rvert \leq \lvert \mb{p}(S)\cdot \mb{c}\rvert$. 
In particular, $\mb{p}(S)\cdot \mb{c}=0$ implies $\mb{p}(S') \cdot \mb{c}=0$. 
Considering another expression of the difference, we deduce that 
for $\mb{d}\geq 0$ satisfying $\mb{r}(S)\cdot \mb{d}=0$ there exists some $S'$ such that $\mb{r}(S') \cdot \mb{d}=0$.  
By regarding these expressions as inequalities for $\mb{d}$, we obtain 
\[
 \{\mb{d}\in \mathbb{Z}^n : \mb{r}(S)\cdot \mb{d}=0, \;\mb{d}\geq \mb{0}\} \subset  
 \bigcup_{S'\neq S, [n]\setminus S}
 \{\mb{d}\in \mathbb{Z}^n : \mb{r}(S')\cdot \mb{d}=0,\;\mb{d}\geq \mb{0}\}. 
\]
We will denote the left-hand side and the set in the union on the right-hand side by $C_S$ and $C_{S'}$ respectively. From the notation and the expression above, we can deduce that $C_S$ is equal to the union of the intersections of $C_S$ with each $C_{S'}$, that is,
\[
 C_S = \bigcup_{S'\neq S, [n]\setminus S} (C_S \cap C_{S'}). 
\]
Now choose a minimal decomposition
\begin{equation}\label{eq:decomp}
C_S = \bigcup_{i=1}^m (C_S \cap C_{S'_i}) 
\end{equation}
with $m$ minimal. We claim that $m \geq 2$. First, $\mb{r}(S')$ is not proportional to $\mb{r}(S)$ for any admissible $S'$: since the first entries of cumulative sums are $\pm 1$, a relation $\mb{r}(S') = \lambda\, \mb{r}(S)$ would force $\lambda = \pm 1$, that is, $S' = S$ or $S' = [n] \setminus S$. Second, write $\mb{r}(S) = (r_1, \dots, r_n)$ and let $\mb{e}_i$ denote the $i$th standard unit vector. Since $\mb{p}(S) \in Q(n)$, the vector $\mb{r}(S)$ has both a positive entry and a negative one, so $C_S$ contains $\mb{e}_i$ for every $i$ with $r_i = 0$ and $\lvert r_j \rvert\, \mb{e}_i + r_i\, \mb{e}_j$ for all $i, j$ with $r_i > 0 > r_j$; these vectors span the hyperplane $\{\mb{d} \in \mathbb{R}^n \colon \mb{r}(S)\cdot \mb{d} = 0\}$. Hence $C_S \not\subset C_{S'}$ for every admissible $S'$, and a single term cannot cover $C_S$ in (\ref{eq:decomp}). By minimality, no $C_S \cap C_{S'_i}$ is contained in the union of other $C_S \cap C_{S'_j}$'s.
Take vectors $\mb{d}_1$ and $\mb{d}_2$ such that $\mb{d}_1 \in C_S \cap C_{S'_1}$ but $\mb{d}_1 \notin C_S \cap C_{S'_i}$ for $i \neq 1$ and $\mb{d}_2 \in C_S \cap C_{S'_2}$ but $\mb{d}_2 \notin C_S \cap C_{S'_i}$ for $i \neq 2$. 
Consider the intersection of $C_S \cap C_{S'_i}$ and the infinite set $K = \{ \alpha \mb{d}_1 + 
 \mb{d}_2 : \alpha \in \mathbb{Z}, \alpha \geq 0\} \subset C_S$. 
If $i=1$, then the intersection is empty since $\mb{r}(S'_i)\cdot \mb{d}_1 = 0$ and $\mb{r}(S'_i)\cdot \mb{d}_2 \neq 0$. 
If $i \neq 1$, then the intersection is at most one point since $\mb{r}(S'_i)\cdot \mb{d}_1 \neq 0$.
So $K$ intersects each $C_S \cap C_{S'_i}$ in at most one point. 
Since (\ref{eq:decomp}) holds, $K$ is a finite set, a contradiction. 
\end{proof}





We thus can prove \Cref{th:sol}.
\begin{proof}[Proof of \Cref{th:sol}]
If $\mb{p}(S) \notin Q(n)$, that is, $\mb{p}(S) \prec \mb{0}$ or $\mb{0} \prec \mb{p}(S)$, then $S$ is dispensable by \Cref{lem_order&d}. If $\mb{p}(S) \in Q(n)$, then \Cref{lem_order&d&r} yields an instance $\mb{c}$ with $\lvert \mb{p}(S)\cdot\mb{c} \rvert < \lvert \mb{p}(S')\cdot\mb{c} \rvert$ for every $S'$ not equal to $S$ or $[n]\setminus S$, so $S$ is not dispensable. 
\end{proof}

We now consider a specific instance $\mb{c}$.
We can prune candidate solutions based on the partial order.
\begin{proposition}\label{prop:sign}
  Given an instance $\mb{c}$, for some subset $S$ of $[n]$, the following statements hold:
  \begin{enumerate}
    \item If $\mb{p}(S)\cdot \mb{c} \geq 0$, 
  then $\lvert \mb{p}(S)\cdot \mb{c}\rvert \leq \lvert \mb{p}(S')\cdot \mb{c}\rvert$ for any $S'$ such that $\mb{p}(S') \preceq_P -\mb{p}(S)$ or $\mb{p}(S) \preceq_P \mb{p}(S')$.
  \item If $\mb{p}(S)\cdot \mb{c} \leq 0$, 
then $\lvert \mb{p}(S)\cdot \mb{c}\rvert \leq \lvert \mb{p}(S')\cdot \mb{c}\rvert$ for any $S'$ such that $\mb{p}(S') \preceq_P \mb{p}(S)$ or $-\mb{p}(S) \preceq_P \mb{p}(S')$.
    \end{enumerate}
\end{proposition}

\begin{proof}
We show the first statement.
From \Cref{prop:d&order} it follows that $\mb{p}(S')\cdot \mb{c} \leq -\mb{p}(S)\cdot \mb{c}$ or $\mb{p}(S)\cdot \mb{c} \leq \mb{p}(S')\cdot \mb{c}$.
Hence, $\mb{p}(S')\cdot \mb{c} \leq -\mb{p}(S)\cdot \mb{c} \leq 0$ or $0 \leq \mb{p}(S)\cdot \mb{c} \leq \mb{p}(S')\cdot \mb{c}$. 
Either case implies that $\lvert \mb{p}(S)\cdot \mb{c}\rvert \leq \lvert \mb{p}(S')\cdot \mb{c}\rvert$.

In the same manner, we can prove the second.  
\end{proof}

We can show that a similar statement, which is \Cref{th:signQ}, holds for $Q(n)$.
\begin{proof}[Proof of \Cref{th:signQ}]
It immediately follows from \Cref{prop:sign}.  
\end{proof}

\begin{remark}
The width of $Q(n)$ is $\Theta(2^n/ n^{3/2})$ for $n$ congruent to 0 or 3 modulo 4. This exponential width indicates the hardness of the partition problem: the instance-independent dominance relation (\Cref{prop:d&order}) holds only between comparable elements, so it orders none of the exponentially many pairwise-incomparable elements of a maximum antichain, which must instead be compared using the instance.

\end{remark}

Having described this global obstruction, we now bound the reduction achievable by a single application of \Cref{th:signQ}.
For $\mb{v} \in Q(n)$, let $u(\mb{v})$ and $d(\mb{v})$ denote the numbers of elements of $Q(n)$ greater than $\mb{v}$ and less than $\mb{v}$, respectively.
Suppose that we test the sign of the single inner product $\mb{v}\cdot \mb{c}$.
If $\mb{v}\cdot \mb{c} \geq 0$, then by \Cref{th:signQ} the $u(\mb{v})$ elements greater than $\mb{v}$ and the $u(\mb{v})$ elements less than $-\mb{v}$ (their negations, by \Cref{prop:symQ}) are discarded in favor of $\mb{v}$; these $2u(\mb{v})$ elements are distinct and different from $\pm\mb{v}$, since $\mb{v} \prec_Q \mb{w}$ and $\mb{w} \prec_Q -\mb{v}$ would imply $\mb{v} \prec \mb{0}$.
If $\mb{v}\cdot \mb{c} \leq 0$, the $2d(\mb{v})$ elements less than $\mb{v}$ or greater than $-\mb{v}$ are discarded likewise.
Whichever sign occurs, at least
\[
\gamma(\mb{v}) = 2 \min \{u(\mb{v}),\, d(\mb{v})\}
\]
candidate solutions other than $\pm\mb{v}$ are discarded.
\Cref{tab:min} shows the maximum of $\gamma(\mb{v})$ over $Q(n)$ together with a maximizing element $\mb{v}^*$ for $5 \leq n \leq 15$; in this range the maximizing pair $(\mb{v}^*, -\mb{v}^*)$ is unique, and $\mb{v}^*$ is an almost alternating vector. The maximum grows exponentially. Nevertheless, no single test is guaranteed to discard a constant fraction of the candidate solutions:

\begin{table}[ht]
\caption{The maximum guaranteed reduction $\gamma(\mb{v}^*)$ of the candidate solutions by a single sign test and a maximizing element $\mb{v}^*$ for $5 \leq n \leq 15$.}
\label{tab:min}
\centering
\begin{tabular}{rr@{\qquad}r@{\qquad}r@{}r@{}r@{}r@{}r@{}r@{}r@{}r@{}r@{}r@{}r@{}r@{}r@{}r@{}r@{}r}
\hline
$n$ & $\#Q(n)$ & $\gamma(\mb{v}^*)$ & \multicolumn{16}{c}{$\mb{v}^*$}\\
\hline
 5& 12& 4& (&1,&$-$1,&$-$1,&1,&$-$1)\\
 6& 24& 8& (&1,&$-$1,&$-$1,&1,&$-$1,&1)\\
 7& 58& 20& (&1,&$-$1,&1,&$-$1,&$-$1,&$-$1,&1)\\
 8& 116& 38 & (&1,&$-$1,&$-$1,&1,&$-$1,&1,&$-$1,&1)\\
 9& 260& 84 & (&1,&$-$1,&1,&$-$1,&$-$1,&1,&$-$1,&$-$1,&1)\\
 10 & 520& 152 & (&1,&$-$1,&1,&$-$1,&$-$1,&1,&$-$1,&$-$1,&1,&$-$1)\\
 11 & 1124& 324 & (&1,& $-$1,& 1,& $-$1,& $-$1,& 1,& $-$1,& 1,& $-$1,& $-$1,& 1)\\
 12 & 2248& 594 & (&1,& $-$1,& 1,& $-$1,& $-$1,& 1,& $-$1,& 1,& $-$1,& $-$1,& 1,& $-$1)\\
 13 & 4760& 1230 & (&1,& $-$1,& 1,& $-$1,& $-$1,& 1,& $-$1,& 1,& $-$1,& 1,& $-$1,& $-$1,& 1)\\
 14 & 9520& 2290 & (&1,& $-$1,& 1,& $-$1,& 1,& $-$1,& $-$1,& $-$1,& 1,& $-$1,& 1,& $-$1,& 1,&$-$1)\\
 15 & 19898& 4666 & (&1,& $-$1,& 1,& $-$1,& $-$1,& 1,& $-$1,& 1,& $-$1,& 1,& $-$1,& 1,& $-$1,& $-$1,& \phantom{$-$}1)\\
\hline
\end{tabular}
\end{table}

\begin{proposition}\label{prop:singlequery}
For every $\mb{v} \in Q(n)$,
\[
\gamma(\mb{v}) \leq \frac{2^{n+2}}{(n+1)^{1/4}}.
\]
In particular, $\gamma(\mb{v})/\#Q(n) = O(n^{-1/4})$: no single sign test can, for example, halve the number of candidate solutions.
\end{proposition}

\begin{proof}
Let $u_P(\mb{v})$ and $d_P(\mb{v})$ denote the sizes of the up-set and the down-set of $\mb{v}$ in $P(n)$ (both containing $\mb{v}$ itself), so that $u(\mb{v}) \leq u_P(\mb{v})$ and $d(\mb{v}) \leq d_P(\mb{v})$.
The product $u_P(\mb{v})\, d_P(\mb{v})$ counts the pairs $(\mb{w}, \mb{w}') \in P(n) \times P(n)$ with $\mb{w}' \preceq_P \mb{v} \preceq_P \mb{w}$. For every such pair the cumulative sums of $\mb{w} - \mb{w}'$ are nonnegative.
The entries of $\mb{w} - \mb{w}'$ lie in $\{2, 0, -2\}$, where the entry $0$ arises in two ways ($(w_i, w'_i) = (1,1)$ or $(-1,-1)$) and the entries $\pm 2$ in one way each ($(1,-1)$ and $(-1,1)$, respectively); the cumulative sums of $\mb{w} - \mb{w}'$ are nonnegative exactly when the $\pm1$-sequence formed by the signs of the nonzero entries has nonnegative cumulative sums. The number of such sequences of length $m$ is $X_m = \binom{m}{\lfloor m/2 \rfloor}$, as shown in the proof of \Cref{prop:Qsize}. Grouping the pairs by the set of indices of the nonzero entries, we obtain
\[
u_P(\mb{v})\, d_P(\mb{v}) \leq \sum_{m=0}^{n} \binom{n}{m}\, 2^{n-m} X_m.
\]
We claim that $X_m \leq 2^m/\sqrt{m+1}$, that is, $c_m \leq 1$ for $c_m = X_m^2 (m+1)/4^m$: indeed, $c_0 = 1$, and the identities $X_{2k} = 2X_{2k-1}$ and $X_{2k+1} = X_{2k}(2k+1)/(2k+2) \cdot 2$ give $c_{2k} = c_{2k-2}(4k^2-1)/(4k^2)$ and $c_{2k+1} = c_{2k}(2k+1)/(2k+2)$, so the even-indexed terms are decreasing and each odd-indexed term is less than its predecessor.
By the claim, the Cauchy--Schwarz inequality, and the identity $\sum_{m=0}^n \binom{n}{m}/(m+1) = (2^{n+1}-1)/(n+1)$,
\[
u_P(\mb{v})\, d_P(\mb{v}) \leq 2^n \sum_{m=0}^{n} \frac{\binom{n}{m}}{\sqrt{m+1}}
\leq 2^n \sqrt{2^n \cdot \frac{2^{n+1}-1}{n+1}}
\leq \frac{\sqrt{2}\; 4^n}{\sqrt{n+1}}.
\]
Therefore
$\gamma(\mb{v}) \leq 2\sqrt{u(\mb{v})\, d(\mb{v})} \leq 2\sqrt{u_P(\mb{v})\, d_P(\mb{v})} \leq 2^{n+5/4}/(n+1)^{1/4} \leq 2^{n+2}/(n+1)^{1/4}$.
The last claim of the proposition follows since $\#Q(n) = 2^n - 2\binom{n}{\lfloor n/2 \rfloor} = (1-o(1))\, 2^n$.
\end{proof}


\subsection{Optimality Criterion and Polynomially Solvable Cases}
At the heart of this section is the optimality criterion of \Cref{th:opt}, a necessary and sufficient condition for an element of $Q(n)$ to attain the optimal value. In general its third condition ranges over all elements incomparable to the witnessing element and its negative, and may involve exponentially many inequalities; for the minimal and maximal elements of $Q(n)$ and the elements near them, however, this set is small, so the criterion collapses to a constant number of inequalities. We emphasize that, in general, \Cref{th:opt} is a structural characterization rather than an algorithm: finding a regular element $\mb{w}$ witnessing the optimality of a given $\mb{v}$ is not obviously easier than the problem itself. The resulting polynomially solvable cases follow as corollaries (\Cref{th:min}, \Cref{cor:maximal}, \Cref{cor:extended}). We first prove \Cref{th:opt}, beginning with some terminology and an auxiliary sufficient condition (\Cref{prop:N}).

For $\mb{v} \in Q(n)$ we write $N_{\rm C}(\mb{v})$ for the union of the up-sets and down-sets of $\mb{v}$ and of $-\mb{v}$, that is, the set of elements comparable to $\mb{v}$ or $-\mb{v}$; its complement in $Q(n)$ is the set of elements incomparable to both $\mb{v}$ and $-\mb{v}$.

\begin{proposition}\label{prop:N}
    Given an instance $\mb{c}$ of the partition problem, let $\mb{v}$ be a non-minimal element of $Q(n)$.
    If $\mb{v}\cdot \mb{c} \geq 0$ and $(\mb{v}+\mb{w})\cdot \mb{c} \leq 0$ for every element $\mb{w}$ of $Q(n)$ covered by $\mb{v}$, then
    $\min_{\mb{x} \in N_{\rm C}(\mb{v})} \lvert \mb{x}\cdot \mb{c} \rvert = \mb{v}\cdot \mb{c}$.
\end{proposition}

\begin{proof}
Let $\mb{x} \in N_{\rm C}(\mb{v})$, not equal to $\pm \mb{v}$. By the symmetry of $Q(n)$ we only have to consider the two cases:  $\mb{v}\prec_Q \mb{x}$ and $\mb{x}\prec_Q \mb{v}$.
If $\mb{v}\prec_Q \mb{x}$, then by \Cref{prop:sign}(1) we have $|\mb{v}\cdot \mb{c}| \leq |\mb{x}\cdot \mb{c}|$. 
On the other hand, if $\mb{x}\prec_Q \mb{v}$, then $\mb{x} \preceq_Q \mb{w}$ for some element $\mb{w}$ covered by $\mb{v}$. The condition $(\mb{v}+\mb{w})\cdot \mb{c} \leq 0$ gives $\mb{w}\cdot \mb{c} \leq - \mb{v}\cdot \mb{c} \leq 0$, hence $|\mb{v}\cdot \mb{c}| \leq |\mb{w}\cdot \mb{c}|$; by \Cref{prop:sign}(2) we have $|\mb{w}\cdot \mb{c}| \leq |\mb{x}\cdot \mb{c}|$. Therefore, we obtain $|\mb{v}\cdot \mb{c}| \leq |\mb{x}\cdot \mb{c}|$.
\end{proof}

We now prove \Cref{th:opt}, the necessary and sufficient condition for optimality stated in \Cref{ssec:main}.

\begin{proof}[Proof of \Cref{th:opt}]
First we prove the necessity. The optimal value is nonnegative, so $\mb{v}\cdot \mb{c} \geq 0$. Suppose $\mb{v}\cdot \mb{c} > 0$. Let
\[
L = \{\mb{u} \in Q(n) : \mb{u} \preceq_Q \mb{v} \text{ and } \mb{u}\cdot \mb{c} = \mb{v}\cdot \mb{c}\},
\]
which is nonempty since $\mb{v} \in L$, and let $\mb{w}$ be a minimal element of $L$. Then $\mb{w}\preceq_Q \mb{v}$ and $\mb{w}\cdot \mb{c} = \mb{v}\cdot \mb{c} > 0$, which is condition (i). Moreover, $\mb{w}$ is regular: if $\mb{u}\prec_Q \mb{w}$, then $\mb{u}\cdot \mb{c} \leq \mb{w}\cdot \mb{c}$ by \Cref{prop:d&order}, and $\mb{u}\cdot \mb{c} = \mb{w}\cdot \mb{c}$ would give $\mb{u}\in L$, contradicting the minimality of $\mb{w}$; hence $\mb{u}\cdot \mb{c} < \mb{w}\cdot \mb{c}$.
Since $\mb{w}\cdot \mb{c}$ is also the optimal value, $\mb{w}\cdot \mb{c} \leq |\mb{w}_1\cdot \mb{c}|$ for every element $\mb{w}_1$ covered by $\mb{w}$. By regularity $\mb{w}_1\cdot \mb{c} < \mb{w}\cdot \mb{c}$; combined with $\mb{w}\cdot \mb{c} \leq |\mb{w}_1\cdot \mb{c}|$ this forces $\mb{w}_1\cdot \mb{c} < 0$, whence $|\mb{w}_1\cdot \mb{c}| = -\mb{w}_1\cdot \mb{c} \geq \mb{w}\cdot \mb{c}$, that is, $(\mb{w}+\mb{w}_1)\cdot \mb{c} \leq 0$, which is condition (ii). Likewise, for every element $\mb{w}_2$ incomparable to both $\mb{w}$ and $-\mb{w}$, optimality gives $\mb{w}\cdot \mb{c} \leq |\mb{w}_2\cdot \mb{c}|$, that is, $\mb{w}\cdot \mb{c}\leq \mb{w}_2\cdot \mb{c}$ or $\mb{w}_2\cdot \mb{c} \leq -\mb{w}\cdot \mb{c}$, which is condition (iii).

Next, we prove the sufficiency. If $\mb{v} \cdot \mb{c} = 0$, then $\mb{v} \cdot \mb{c}$ is the optimal value, since the optimal value is nonnegative. Suppose now that there exists a regular element $\mb{w}$ as in the statement, satisfying conditions (i)--(iii).
We show that $\mb{w}\cdot \mb{c} \leq \lvert\mb{x}\cdot \mb{c}\rvert$ for every $\mb{x} \in Q(n)$.
There are two cases:

\noindent \textbf{Case 1:}
$\mb{x}$ is in $N_{\rm C}(\mb{w})$, that is, $\mb{x}$ is comparable to $\mb{w}$ or $-\mb{w}$.
Then condition (ii) and \Cref{prop:N} give $\mb{w}\cdot \mb{c} \leq \lvert\mb{x}\cdot \mb{c}\rvert$. (When $\mb{w}$ is minimal, condition (ii) is vacuous; in this case \Cref{lem:compara} shows that $\mb{w} \preceq_Q \mb{x}$ or $\mb{x} \preceq_Q -\mb{w}$ for every $\mb{x} \in Q(n)$, and \Cref{th:signQ} gives the same conclusion.)

\noindent \textbf{Case 2:}
$\mb{x}$ is not in $N_{\rm C}(\mb{w})$, that is, $\mb{x}$ is incomparable to the two elements.
It follows from condition (iii) that $(\mb{w} + \mb{x})\cdot \mb{c} \leq 0$ or  $(\mb{w} - \mb{x})\cdot \mb{c} \leq 0$, which means that $\mb{w}\cdot \mb{c} \leq \lvert \mb{x}\cdot \mb{c}\rvert$.

Combining the two cases, $\mb{v}\cdot \mb{c} = \mb{w}\cdot \mb{c} \leq \lvert\mb{x}\cdot \mb{c}\rvert$ for every $\mb{x} \in Q(n)$; since the value $\mb{v}\cdot \mb{c} = \lvert\mb{v}\cdot \mb{c}\rvert$ is attained, it is the optimal value.
\end{proof}

We now specialize \Cref{th:opt} to the extreme elements of $Q(n)$, beginning with the minimal ones.

\begin{proof}[Proof of \Cref{th:min}]
Apply \Cref{th:opt} to $\mb{v} = -\mb{m}_k$. Since $-\mb{m}_k$ is minimal, the only element $\mb{w} \preceq_Q \mb{v}$ is $\mb{v}$ itself, which is regular (it has no element below it) and covers no elements; moreover, by \Cref{lem:compara} every element of $Q(n)$ is comparable to $-\mb{m}_k$ or $\mb{m}_k$, so none is incomparable to both. Thus, for the witness $\mb{w} = \mb{v}$, conditions (ii) and (iii) of \Cref{th:opt} hold vacuously. Consequently, if $-\mb{m}_k\cdot\mb{c}\geq 0$, then either $-\mb{m}_k\cdot\mb{c}=0$ or condition (i) holds, and \Cref{th:opt} shows that $-\mb{m}_k\cdot\mb{c}$ is the optimal value. Conversely, if $-\mb{m}_k\cdot\mb{c}$ is the optimal value, then $-\mb{m}_k\cdot\mb{c}\geq 0$ because the optimal value is nonnegative.
\end{proof}



The maximal elements, like the minimal ones, have no elements incomparable to both themselves and their negatives (\Cref{lem:compara}); this yields \Cref{cor:maximal}.

\begin{proof}[Proof of \Cref{cor:maximal}]
First let $n \geq 5$. By \Cref{lem:compara}, every element of $Q(n)$ is comparable to $\mb{m}_k$ or $-\mb{m}_k$, so none is incomparable to both. The elements covering $-\mb{m}_k$ in $P(n)$ are $\mathcal{S}^{(k,k+1)}(-\mb{m}_k)$ for $k\neq 0$ and $\mathcal{A}^{(n)}(-\mb{m}_k)$ for $k\neq(n-1)/2$ (\Cref{prop:cover}); for $n \geq 5$ their sequences of cumulative sums attain both signs, so they lie in $Q(n)$ and, by \Cref{cor:coverQ}, they are exactly the elements covering $-\mb{m}_k$ in $Q(n)$. By the symmetry of $Q(n)$ (\Cref{prop:symQ}), the elements covered by $\mb{m}_k$ are their negatives. Thus, with $\mb{m}_k$ in place of $\mb{w}$, the inequalities (i)--(iii) coincide with conditions (i)--(ii) of \Cref{th:opt}, and condition (iii) is vacuous.

If (i)--(iii) hold, then \Cref{prop:N} gives $\mb{m}_k\cdot\mb{c}\leq\lvert\mb{x}\cdot\mb{c}\rvert$ for every $\mb{x}\in Q(n)$, so $\mb{m}_k\cdot\mb{c}$ is the optimal value.

For $n \in \{3, 4\}$ some of the vectors $\mathcal{S}^{(k,k+1)}(-\mb{m}_k)$ and $\mathcal{A}^{(n)}(-\mb{m}_k)$ fall outside $Q(n)$, and the corresponding inequalities are merely additional requirements, which cannot hurt sufficiency. Since $Q(3) = \{\pm\mb{m}_0\}$ (where $\mb{m}_1 = -\mb{m}_0$) and $Q(4) = \{\pm\mb{m}_0, \pm\mb{m}_1\}$, the claim is verified directly; for instance, for $n=4$ and $k=1$, condition (ii) holds automatically (it reads $c_4 \leq c_3$), while conditions (i) and (iii) read $\mb{m}_1\cdot\mb{c} \geq 0$ and $c_1 \leq c_2 + c_3$, and the latter gives $\mb{m}_1\cdot\mb{c} \leq c_4 \leq \lvert \mb{m}_0\cdot\mb{c} \rvert$.
\end{proof}

In general, condition (iii) of \Cref{th:opt} contains a large number of inequalities and it is difficult to write down the condition. If the difference between $Q(n)$ and $N_{\rm C}(\mb{v})$ is small, however, the condition contains only a few inequalities and can be written down explicitly, leading to further polynomially solvable cases. We next investigate this difference for the covering elements of the minimal elements. Even the covers of the minimal elements---arguably the simplest non-minimal elements of $Q(n)$---are already incomparable to exponentially many elements. The next two propositions count them exactly for the two cover types, the swap covers $\mathcal{S}^{(k,k+1)}(-\mb{m}_k)$ and the addition covers $\mathcal{A}^{(n)}(-\mb{m}_k)$, which behave oppositely in $k$; \Cref{rem:asymp} draws the asymptotic consequences.

For $\mb{x} \in P(n)$ we write $\sigma_i(\mb{x}) = x_1 + \cdots + x_i$ for the $i$th entry of its cumulative sum (so that $\mb{r}(S) = (\sigma_1(\mb{p}(S)), \dots, \sigma_n(\mb{p}(S)))$), so that $\mb{x} \preceq \mb{y}$ if and only if $\sigma_i(\mb{x}) \leq \sigma_i(\mb{y})$ for all $i$; recall $\sigma_{2k+1}(-\mb{m}_k) = 1$.

\begin{proposition}\label{prop:countS}
For $k \in \{2, \dots, \ell\}$, the number of elements of $Q(n)$ incomparable to both $\mathcal{S}^{(k,k+1)}(-\mb{m}_k)$ and its negative is $2^{n-2k}-2$.
\end{proposition}

\begin{proof}
Write $\mb{v} = \mathcal{S}^{(k,k+1)}(-\mb{m}_k)$, let $I$ be the set of elements incomparable to both $\mb{v}$ and $-\mb{v}$, and put $U = \{\mb{x} \in Q(n) : -\mb{m}_k \preceq_Q \mb{x}\}$. By \Cref{prop:symQ}, negation is an order-reversing bijection of $Q(n)$ that maps $I$ to itself and $U$ to $L := \{\mb{x} \in Q(n) : \mb{x} \preceq_Q \mb{m}_k\}$. Every $\mb{x} \in Q(n)$ lies in $U$ or $L$ (\Cref{lem:compara}), and $U \cap L = \emptyset$ since $-\mb{m}_k \not\preceq_Q \mb{m}_k$: the relation $-\mb{m}_k \preceq_Q \mb{m}_k$ would force $-\mb{m}_k \preceq \mb{0}$, contradicting $-\mb{m}_k \in Q(n)$. Hence $\#I = 2\,\#(I \cap U)$.

The swap $\mathcal{S}^{(k,k+1)}$ raises $\sigma_k$ by $2$ and leaves all other cumulative sums unchanged, so $\mb{v}$ and $-\mb{m}_k$ differ only in that $\sigma_k(\mb{v}) = -k+2$. Thus for $\mb{x} \in U$,
\[
\mb{x} \not\succeq_Q \mb{v} \Longleftrightarrow \sigma_k(\mb{x}) < -k+2 \Longleftrightarrow \sigma_k(\mb{x}) = -k,
\]
the last step because $\sigma_k(\mb{x}) \geq \sigma_k(-\mb{m}_k) = -k$ and $\sigma_k(\mb{x}) \equiv k \pmod 2$. Now $\sigma_k(\mb{x}) = -k$ forces $x_1 = \cdots = x_k = -1$. By the description of the up-set of $-\mb{m}_k$ in the proof of \Cref{lem:compara}---write $\mb{x} = (\mb{x}_1, \mb{x}_2)$ with $\mb{x}_1 \in P(2k+1)$ having at most $k$ entries equal to $-1$ and $\mb{x}_2 \in P(n-2k-1)$ arbitrary---this forces $\mb{x}_1 = (\underbrace{-1, \ldots, -1}_{k}, \underbrace{1, \ldots, 1}_{k+1})$ and leaves $\mb{x}_2$ free; each such $\mb{x}$ lies in $Q(n)$, having $\sigma_k(\mb{x}) = -k < 0$ and $\sigma_{2k+1}(\mb{x}) = 1 > 0$. Hence $\#\{\mb{x} \in U : \mb{x} \not\succeq_Q \mb{v}\} = 2^{\,n-2k-1}$.

This set contains $-\mb{m}_k$, which is less than $\mb{v}$; every other member $\mb{x}$ is incomparable to both $\mb{v}$ and $-\mb{v}$. It is incomparable to $\mb{v}$: here $\mb{x} \not\succeq_Q \mb{v}$, and $-\mb{m}_k$ is the only element of $U$ below $\mb{v}$, so $\mb{x} \neq -\mb{m}_k$ gives $\mb{x} \not\preceq_Q \mb{v}$. It is incomparable to $-\mb{v}$: since $k \geq 2$ we have $(-\mb{v})_1 = 1$, so $\sigma_1(\mb{x}) = -1 < 1 = \sigma_1(-\mb{v})$ gives $\mb{x} \not\succeq_Q -\mb{v}$, while $\sigma_{2k+1}(-\mb{v}) = -\sigma_{2k+1}(\mb{v}) = -1 < 1 = \sigma_{2k+1}(\mb{x})$ gives $\mb{x} \not\preceq_Q -\mb{v}$. Conversely, any $\mb{x} \in U$ incomparable to $\mb{v}$ satisfies $\mb{x} \not\succeq_Q \mb{v}$, hence belongs to this set and differs from $-\mb{m}_k$. Therefore $\#(I \cap U) = 2^{\,n-2k-1} - 1$ and $\#I = 2^{\,n-2k} - 2$.
\end{proof}

\begin{proposition}\label{prop:countA}
For $k \in \{0, 1, \dots, \lfloor (n-4)/2 \rfloor\}$, the number of elements of $Q(n)$ incomparable to both $\mathcal{A}^{(n)}(-\mb{m}_k)$ and its negative is $\binom{2k+2}{k+1} - 2$.
\end{proposition}

\begin{proof}
Write $\mb{v} = \mathcal{A}^{(n)}(-\mb{m}_k)$ and keep the notation $I$, $U$, $L$ of the previous proof; as there, $U \cap L = \emptyset$ and $\#I = 2\,\#(I \cap U)$.

The operator $\mathcal{A}^{(n)}$ raises only the last entry, so $\sigma_i(\mb{v}) = \sigma_i(-\mb{m}_k)$ for $i < n$ while $\sigma_n(\mb{v}) = \sigma_n(-\mb{m}_k) + 2$. Hence for $\mb{x} \in U$,
\[
\mb{x} \not\succeq_Q \mb{v} \Longleftrightarrow \sigma_n(\mb{x}) = \sigma_n(-\mb{m}_k) = 2k+2-n,
\]
that is, $\mb{x}$ has exactly $k+1$ entries equal to $1$, as does $-\mb{m}_k$. Writing $\mb{x} = (\mb{x}_1, \mb{x}_2)$ as above, the vector $\mb{x}_1 \in P(2k+1)$ has at most $k$ entries equal to $-1$, since $-\mb{m}_k \preceq_Q \mb{x}$ gives $\sigma_{2k+1}(\mb{x}) \geq \sigma_{2k+1}(-\mb{m}_k) = 1$; that is, $\mb{x}_1$ has at least $k+1$ entries equal to $1$. Hence the total number of $1$'s is $k+1$ exactly when $\mb{x}_1$ has $k+1$ $1$'s and $\mb{x}_2 = (-1, \ldots, -1)$. Conversely, every $\mb{x}_1 \in P(2k+1)$ with $k+1$ $1$'s satisfies $\mb{x}_1 \succeq (\underbrace{-1, \ldots, -1}_{k}, \underbrace{1, \ldots, 1}_{k+1})$, since placing all $-1$'s first minimizes the cumulative sums; the resulting $\mb{x}$ then satisfies $-\mb{m}_k \preceq \mb{x}$ and lies in $Q(n)$, having $\sigma_{2k+1}(\mb{x}) = 1 > 0$ and $\sigma_n(\mb{x}) = 2k+2-n \leq -2 < 0$, so $\mb{x} \in U$. Thus $\#\{\mb{x} \in U : \mb{x} \not\succeq_Q \mb{v}\} = \binom{2k+1}{k+1} = \binom{2k+1}{k}$.

This set contains $-\mb{m}_k$, which is less than $\mb{v}$; every other member $\mb{x}$ is incomparable to both $\mb{v}$ and $-\mb{v}$. It is incomparable to $\mb{v}$: here $\mb{x} \not\succeq_Q \mb{v}$, and $-\mb{m}_k$ is the only element of $U$ below $\mb{v}$, so $\mb{x} \neq -\mb{m}_k$ gives $\mb{x} \not\preceq_Q \mb{v}$. It is incomparable to $-\mb{v}$: from $\sigma_{2k+1}(-\mb{v}) = -\sigma_{2k+1}(\mb{v}) = -1$ and $\sigma_{2k+1}(\mb{x}) = 1$ we get $\mb{x} \not\preceq_Q -\mb{v}$, while $\sigma_n(-\mb{v}) = -\sigma_n(\mb{v}) = n-2k-4 > 2k+2-n = \sigma_n(\mb{x})$—the inequality holding since $n-2k-4 \geq 0$ by the hypothesis $k \leq \lfloor(n-4)/2\rfloor$—gives $\mb{x} \not\succeq_Q -\mb{v}$. Conversely, any $\mb{x} \in U$ incomparable to $\mb{v}$ satisfies $\mb{x} \not\succeq_Q \mb{v}$, hence belongs to this set and differs from $-\mb{m}_k$. Therefore $\#(I \cap U) = \binom{2k+1}{k} - 1$, and since $\binom{2k+2}{k+1} = 2\binom{2k+1}{k}$ we obtain $\#I = \binom{2k+2}{k+1} - 2$.
\end{proof}

\begin{remark}
The hypotheses are needed: for $\mathcal{S}^{(k,k+1)}(-\mb{m}_k)$ with $k = 1$, and for $\mathcal{A}^{(n)}(-\mb{m}_k)$ at $k = \ell$ when $n$ is even and at $k = \ell-1$ when $n$ is odd, some of the elements counted above become comparable to $-\mb{v}$ or leave $Q(n)$ altogether, and the formulas change. In the boundary case $k=1$ of \Cref{prop:countS}, the count is, for $n \geq 5$,
\[
2^{\,n-2} - 2n + 2,
\]
the deficit $2n-4$ from $2^{n-2k}-2$ arising as follows. For $k=1$ one has $\mb{v} = (1,-1,1,-1,\ldots,-1)$ and $-\mb{v} = (-1,1,-1,1,\ldots,1)$, so $(-\mb{v})_1 = -1$ and $\sigma_1(\mb{x}) = -1$ no longer forces $\mb{x} \not\succeq_Q -\mb{v}$. Among the elements counted in the proof of \Cref{prop:countS}---those of the form $\mb{x} = (-1, 1, 1, \mb{x}_2)$ with $\mb{x}_2 \in P(n-3)$ free---exactly those whose tail $\mb{x}_2$ contains at most one $-1$ now satisfy $\mb{x} \succeq_Q -\mb{v}$ and hence become comparable to $-\mb{v}$; there are $1 + (n-3) = n-2$ such elements and, together with their negatives, they account for the deficit $2(n-2) = 2n-4$. In particular a single cover of the minimal element $-\mb{m}_1$ is incomparable to $\Theta(2^n)$ elements---asymptotically a quarter of $Q(n)$.

For the addition cover, the bound $\lfloor(n-4)/2\rfloor$ on $k$ in \Cref{prop:countA} equals $\ell-1$ for even $n$ and $\ell-2$ for odd $n$; since $\mathcal{A}^{(n)}(-\mb{m}_k)$ exists exactly for $k \neq (n-1)/2$, the only addition covers not treated by \Cref{prop:countA} are the two at the largest admissible indices: $k = \ell-1$ for $n$ odd and $k = \ell$ for $n$ even. There the counts are, respectively,
\[
\binom{2\ell}{\ell} - 2\ell - 2 \qquad\text{and}\qquad 2\binom{2\ell+1}{\ell-1} - 2.
\]
Indeed, for $n$ odd and $k = \ell-1$ one has $\sigma_n(-\mb{v}) = \sigma_n(\mb{x}) = -1$, so the strict inequality $\sigma_n(-\mb{v}) > \sigma_n(\mb{x})$ in the proof of \Cref{prop:countA} degenerates to an equality; the other comparison $\sigma_{2k+1}(-\mb{v}) = -1 < 1 = \sigma_{2k+1}(\mb{x})$ is unaffected, so $\mb{x} \not\preceq_Q -\mb{v}$ still holds, and exactly $\ell$ of the counted elements---those whose first $\ell-1$ entries are $1$ with a single further $1$ among the next $\ell$ positions---now satisfy $\mb{x} \succeq_Q -\mb{v}$ and drop out; together with their negatives, this reduces $\binom{2\ell}{\ell}-2$ by $2\ell$.
For $n$ even and $k = \ell$ one has $\sigma_n(\mb{v}) = 2$, so the elements $\mb{x}$ of the up-set of $-\mb{m}_\ell$ in $P(n)$ with $\mb{x} \not\succeq_P \mb{v}$ are those with $\sigma_n(\mb{x}) = 0$, namely the $\binom{2\ell+1}{\ell}$ vectors $(\mb{x}_1, -1)$ with $\mb{x}_1 \in P(2\ell+1)$ containing exactly $\ell$ entries equal to $-1$. Here the correction is different: such an $\mb{x}$ lies outside $Q(n)$ exactly when all its cumulative sums are nonnegative---when its lattice path is one of the $C_{\ell+1}$ paths of length $2\ell+2$ from $0$ to $0$ never going below $0$, as in the proof of \Cref{prop:Qsize}---and these excluded vectors include every $\mb{x}$ with $\mb{x} \succeq_P -\mb{v}$. Each of the remaining $\binom{2\ell+1}{\ell} - C_{\ell+1} = \binom{2\ell+1}{\ell-1}$ vectors is incomparable to $-\mb{v}$: a negative cumulative sum among its first $2\ell$ entries, where those of $-\mb{v}$ are nonnegative, gives $\mb{x} \not\succeq_Q -\mb{v}$, and $\sigma_{2k+1}(-\mb{v}) = -1 < 1 = \sigma_{2k+1}(\mb{x})$ gives $\mb{x} \not\preceq_Q -\mb{v}$ as before. Arguing as in the proof of \Cref{prop:countA} for the comparisons with $\mb{v}$, all of them except $-\mb{m}_\ell$ itself are incomparable to both $\mb{v}$ and $-\mb{v}$, and their negatives double the count to $2\bigl(\binom{2\ell+1}{\ell-1} - 1\bigr)$. Together with \Cref{prop:countA} this yields the incomparable count of every addition cover.
\end{remark}

\begin{remark}\label{rem:asymp}
The two counts of \Cref{prop:countS,prop:countA} vary oppositely with $k$: the swap count $2^{n-2k}-2$ is $\Theta(2^n)$ for small $k$ and drops to $0$ (for $n$ odd) or $2$ (for $n$ even) at $k = \ell$, while the addition count $\binom{2k+2}{k+1}-2 = \Theta(4^k/\sqrt{k})$ is $0$ at $k = 0$ and grows to $\Theta(2^n/\sqrt{n})$ as $k \to \ell$. Hence every minimal element $-\mb{m}_k$ with $1 \leq k \leq \ell-1$ has a cover incomparable to at least $2^{(1/2-o(1))n}$ elements (the larger of the two counts, which balance near $k = n/4$), and so does $-\mb{m}_\ell$ for $n$ even, through the addition cover $\mathcal{A}^{(n)}(-\mb{m}_\ell)$. For each fixed $k \geq 1$, the swap cover alone is incomparable to a $2^{-2k}(1-o(1))$ fraction of $Q(n)$---a fraction independent of $n$. The covers whose incomparable count stays bounded as $n \to \infty$ are the addition covers $\mathcal{A}^{(n)}(-\mb{m}_k)$ with $k$ fixed and the swap covers with $\ell - k$ fixed; among them, only $\mathcal{A}^{(n)}(-\mb{m}_0)$ and $\mathcal{S}^{(\ell,\ell+1)}(-\mb{m}_\ell)$ are incomparable to at most two elements, and their neighborhoods are recorded in rows (2)--(6) of \Cref{Tab:compara} and exploited in \Cref{cor:extended}; the many-incomparable cover $\mathcal{A}^{(n)}(-\mb{m}_\ell)$ (for $n$ even) is accordingly absent from the first column of \Cref{Tab:compara}, appearing instead as the incomparable pair in row (6). This is a local counterpart of the global width bound $\Theta(2^n/n^{3/2})$, though it measures a different quantity: even for the simplest non-minimal elements, condition~(iii) of the optimality criterion (\Cref{th:opt}) already comprises exponentially many inequalities. (The quantity here is the number of elements incomparable to a single $\mb{v}$, i.e.\ its incomparability degree, not the size of an antichain; the two need not coincide, and indeed $2^{n-2}$ far exceeds the width.)
\end{remark}

The following lemma collects, for the minimal and maximal elements and the elements near them, the data needed to evaluate the optimality criterion of \Cref{th:opt} (and hence to apply \Cref{cor:extended}): the covered elements (used in condition (ii)) and the elements incomparable to both $\mb{v}$ and $-\mb{v}$ (used in condition (iii)).
\begin{lemma}\label{lem:table}
Let $n \geq 6$ and $k \in \{0, 1, \ldots, \ell\}$. For an element $\mb{v} \in Q(n)$ as in the first column of \Cref{Tab:compara},
the elements of $Q(n)$ covered by $\mb{v}$ and the elements of $Q(n)$ incomparable to both $\mb{v}$ and $-\mb{v}$ are given by the second column and the third one of the table.
\end{lemma}
\begin{table}[ht]
\caption{Some structural properties of $Q(n)$.}
\label{Tab:compara}
\small
\setlength{\tabcolsep}{4pt}
  \begin{tabular}{rccc}
    \toprule
    &$\mb{v}$ & the elements covered by $\mb{v}$ & the elements \\
    && & incomparable to $\pm \mb{v}$\\
    \midrule
    (1)&$-\mb{m}_k $ & none & none \\ \hline
    (2)&$\mathcal{A}(-\mb{m}_0)$& $-\mb{m}_0$ & none\\ \hline
    (3)&$\mathcal{S}^{(n-1)}\mathcal{A}(-\mb{m}_0)$
     & $\mathcal{A}(-\mb{m}_0)$ & none\\ \hline
    (4)&$\mathcal{S}^{(n-2)}\mathcal{S}^{(n-1)}\mathcal{A}(-\mb{m}_0)$
     & $\mathcal{S}^{(n-1)}\mathcal{A}(-\mb{m}_0)$ 
     & $\pm \mathcal{A}\mathcal{S}^{(n-1)}\mathcal{A}(-\mb{m}_0)$\\ \hline
    (5)&$\mathcal{S}^{(\ell)}(-\mb{m}_\ell)$, $n$: odd
     & $-\mb{m}_\ell$ & none\\ \hline
    (6)&$\mathcal{S}^{(\ell)}(-\mb{m}_\ell)$, $n$: even
     & $-\mb{m}_\ell$ & $\pm \mathcal{A}(-\mb{m}_\ell)$\\ \hline
    (7)&$\mb{m}_k $ & $-\mathcal{S}^{(k)}(-\mb{m}_k), k\neq 0$  & none \\
     && $-\mathcal{A}(-\mb{m}_k), k\neq (n-1)/2$ & \\ \hline
    (8)&$-\mathcal{A}(-\mb{m}_0)$
     & $-\mathcal{S}^{(n-1)}\mathcal{A}(-\mb{m}_0)$ & none\\ \hline
    (9)&$-\mathcal{S}^{(n-1)}\mathcal{A}(-\mb{m}_0)$
     & $-\mathcal{S}^{(n-2)}\mathcal{S}^{(n-1)}\mathcal{A}(-\mb{m}_0)$ & none\\
     && $-\mathcal{A}\mathcal{S}^{(n-1)}\mathcal{A}(-\mb{m}_0)$ & \\ \hline
    (10)&$\quad-\mathcal{S}^{(n-2)}\mathcal{S}^{(n-1)}\mathcal{A}(-\mb{m}_0)$
     & $\quad -\mathcal{S}^{(n-3)}\mathcal{S}^{(n-2)}\mathcal{S}^{(n-1)}\mathcal{A}(-\mb{m}_0)$ 
     & $\quad\pm \mathcal{A}\mathcal{S}^{(n-1)}\mathcal{A}(-\mb{m}_0)$\\
     && $-\mathcal{A}\mathcal{S}^{(n-2)}\mathcal{S}^{(n-1)}\mathcal{A}(-\mb{m}_0)$ & \\ \hline
    (11)&$-\mathcal{S}^{(\ell)}(-\mb{m}_\ell)$, $n$: odd
     & $-\mathcal{S}^{(\ell-1)}\mathcal{S}^{(\ell)}(-\mb{m}_\ell)$ & none\\
     && $-\mathcal{S}^{(\ell+1)}\mathcal{S}^{(\ell)}(-\mb{m}_\ell)$& \\ \hline
    (12)&$-\mathcal{S}^{(\ell)}(-\mb{m}_\ell)$, $n$: even
     & $-\mathcal{S}^{(\ell-1)}\mathcal{S}^{(\ell)}(-\mb{m}_\ell)$ & $\pm \mathcal{A}(-\mb{m}_\ell)$\\
     && $-\mathcal{S}^{(\ell+1)}\mathcal{S}^{(\ell)}(-\mb{m}_\ell)$ & \\
     && $-\mathcal{A}\mathcal{S}^{(\ell)}(-\mb{m}_\ell)$ & \\
    \bottomrule
  \end{tabular}
  \par\smallskip
  {\footnotesize Note: the operators $\mathcal{A}^{(n)}$ and $\mathcal{S}^{(j, j+1)}$ are abbreviated as $\mathcal{A}$ and $\mathcal{S}^{(j)}$, respectively.}
\end{table}

\begin{figure}[ht]
\centering
\begin{tikzpicture}
  \path[] (-1,1) grid (11, 4);
  \coordinate(v0)at(0,3);
  \coordinate(v1)at(2,3);
  \coordinate(v2)at(4,3);

  \coordinate(v4) at(6,3);
  \coordinate(v3) at(6,2);

  \coordinate(v5) at(8,3);
  \coordinate(v6) at(8,2);
\foreach \i/\j in {
0/1,1/2,2/3,2/4,3/6,4/5,4/6
}{
  \draw[gray](v\i)--(v\j);
  }
\foreach \i in {5,6}
{
  \draw[gray, dashed](v\i)--++(1, 0);
}
\foreach \i in {0, 1, 2, 3, 4, 5, 6
}{
     \fill[](v\i) circle(0.06);
     }
 \foreach \i/\l/\a in {
 0/{$-\mb{m}_0$\\$=(\dots, -1,-1,-1,-1)$}/0,
 1/{$=(\dots, -1,-1,-1,1)$\\$\mathcal{A}(-\mb{m}_0)$}/1,
 2/{$\mathcal{S}^{(n-1)}\mathcal{A}(-\mb{m}_0)$\\$=(\dots, -1,-1,1,-1)$}/0,
 4/{$=(\dots, -1,1,-1,-1)$\\$\mathcal{S}^{(n-2)}\mathcal{S}^{(n-1)}\mathcal{A}(-\mb{m}_0)$}/1,
 5/{$=(\dots, 1,-1,-1,-1)$\\$\mathcal{S}^{(n-3)}\mathcal{S}^{(n-2)}\mathcal{S}^{(n-1)}\mathcal{A}(-\mb{m}_0)$}/3,
 6/{$\mathcal{A}\mathcal{S}^{(n-2)}\mathcal{S}^{(n-1)}\mathcal{A}(-\mb{m}_0)$\\
 $=\mathcal{S}^{(n-2)}\mathcal{A}\mathcal{S}^{(n-1)}\mathcal{A}(-\mb{m}_0)$\\
 $=(\dots, -1,1,-1,1)$}/2,
 3/{$\mathcal{A}\mathcal{S}^{(n-1)}\mathcal{A}(-\mb{m}_0)$\\$=(\dots, -1,-1,1,1)$}/0
}{
  \ifnum \a = 1
   \node[above,align=center,font=\tiny]at(v\i){\l};
  \fi
  \ifnum \a = 0
    \node[below,align=center,font=\tiny]at(v\i){\l};
  \fi
  \ifnum \a = 2
     \node[below right,align=left,font=\tiny]at($(v\i)+(-0.5, 0)$){\l};
  \fi
  \ifnum \a = 3
      \node[above right,align=left,font=\tiny]at($(v\i)+(-0.5, 0)$){\l};
  \fi
 }
\end{tikzpicture}\\[1ex]
{\footnotesize (a) near the minimal element $-\mb{m}_0$}\\[3ex]
\begin{tikzpicture}
  \path[] (-1,2) grid (12,5);
  \coordinate(v0)at(0,3);
  \coordinate(v1)at(2,3);
  \coordinate(v2)at(4,3);

  \coordinate(v4) at(6,4);
  \coordinate(v3) at(6,3);
\foreach \i/\j in {
1/2,2/3,2/4
}{
  \draw[gray](v\i)--(v\j);
  }
\foreach \i in {3,4}
{
  \draw[gray, dashed](v\i)--++(1, 0);
}
\foreach \i in {1,2, 3, 4
}{
     \fill[](v\i) circle(0.06);
     }
 \foreach \i/\l/\a in {
 1/{$= (-1,\ldots,-1, -1, 1, 1,\ldots, 1)$\\$-\mb{m}_\ell$}/1,
 2/{$\mathcal{S}^{(\ell)}(-\mb{m}_\ell)$\\$=(-1,\ldots, -1,1, -1, 1\ldots, 1)$}/0,
 3/{$\mathcal{S}^{(\ell+1)}\mathcal{S}^{(\ell)}(-\mb{m}_\ell)$\\
 $=(-1,\ldots, -1,1, 1, -1\ldots, 1)$\\}/2,
 4/{$=(-1,\ldots, 1,-1, -1, 1\ldots, 1)$\\
 $\mathcal{S}^{(\ell-1)}\mathcal{S}^{(\ell)}(-\mb{m}_\ell)$}/3
}{
  \ifnum \a = 1
   \node[above, font=\tiny, align=center]at(v\i){\l};
  \fi
  \ifnum \a = 0
    \node[below, font=\tiny, align=center]at(v\i){\l};
  \fi
  \ifnum \a = 2
     \node[below right, font=\tiny, align=left]at(v\i){\l};
  \fi
  \ifnum \a = 3
      \node[above right, font=\tiny, align=left]at(v\i){\l};
  \fi
 }
\end{tikzpicture}\\[1ex]
{\footnotesize (b) near $-\mb{m}_\ell$, $n$ odd}\\[3ex]
\begin{tikzpicture}
  \path[] (-1,1) grid (12,5);
  \coordinate(v0)at(0,3);
  \coordinate(v1)at(2,3);
  \coordinate(v2)at(2,3);

  \coordinate(v4) at(4,3);
  \coordinate(v3) at(4,2);

  \coordinate(v7) at(6,4);
  \coordinate(v5) at(6,3);
  \coordinate(v6) at(6,2);
\foreach \i/\j in {
2/3,2/4,3/6,4/5,4/6,4/7
}{
  \draw[gray](v\i)--(v\j);
  }
\foreach \i in {5,6,7}
{
  \draw[gray, dashed](v\i)--++(1, 0);
}
\foreach \i in { 2, 3, 4, 5, 6, 7
}{
     \fill[](v\i) circle(0.06);
     }
 \foreach \i/\l/\a in {
  2/{$=(-1,\dots, -1,-1, 1,1 \dots, 1, -1)$\\$-\mb{m}_\ell$}/1,
  4/{$\mathcal{S}^{(\ell)}(-\mb{m}_\ell)$\\
  $=(-1, \dots,-1, 1, -1, 1,\dots, 1, -1)$}/0,
  3/{$\mathcal{A}(-\mb{m}_\ell)$\\$=(-1, \dots, -1,-1, 1, 1\dots, 1,1)$}/0,
  7/{$=(-1, \dots,1, -1, -1, 1,\dots, 1, -1)$\\
  $\mathcal{S}^{(\ell-1)}\mathcal{S}^{(\ell)}(-\mb{m}_\ell)$}/3,
  5/{$=(-1, \dots,-1, 1, 1, -1,\dots, 1, -1)$\\
  $\mathcal{S}^{(\ell+1)}\mathcal{S}^{(\ell)}(-\mb{m}_\ell)$}/3,
  6/{$\mathcal{A}\mathcal{S}^{(\ell)}(-\mb{m}_\ell)$\\
   $=\mathcal{S}^{(\ell)}\mathcal{A}(-\mb{m}_\ell)$\\
   $=(-1, \dots,-1, 1, -1, 1,\dots, 1, 1)$}/2
}{
  \ifnum \a = 1
   \node[above, font=\tiny,align=center]at(v\i){\l};
  \fi
  \ifnum \a = 0
    \node[below, font=\tiny, align=center]at(v\i){\l};
  \fi
  \ifnum \a = 2
     \node[below right, font=\tiny, align=left]at(v\i){\l};
  \fi
  \ifnum \a = 3
      \node[above right, font=\tiny, align=left]at(v\i){\l};
  \fi
 }
\end{tikzpicture}\\[1ex]
{\footnotesize (c) near $-\mb{m}_\ell$, $n$ even}
\caption{The local Hasse structures used in the proof of \Cref{lem:table} ($\mathcal{S}^{(j)}=\mathcal{S}^{(j,j+1)}$, $\mathcal{A}=\mathcal{A}^{(n)}$); only the coordinates near the affected positions are shown, and dashed edges continue upward.}
\label{fig:local}
\end{figure}

\begin{proof}
Throughout we abbreviate $\mathcal{A} := \mathcal{A}^{(n)}$ and $\mathcal{S}^{(j)} := \mathcal{S}^{(j, j+1)}$, as in the table, and recall that
\[
-\mb{m}_k = (\underbrace{-1, \ldots, -1}_{k}, \underbrace{1, \ldots, 1}_{k+1}, \underbrace{-1, \ldots, -1}_{n-2k-1}).
\]
A vector $\mb{w} \in P(n)$ belongs to $Q(n)$ if and only if its sequence of cumulative sums attains both a positive and a negative value.

\smallskip
\noindent\emph{\textup{(I)} Reduction of covers to $P(n)$.}
By \Cref{cor:coverQ}, the covers in $Q(n)$ are exactly the covers in $P(n)$ whose two elements both lie in $Q(n)$. By \Cref{prop:cover}, the elements covering $\mb{v}$ in $P(n)$ are obtained by applying a single operator---$\mathcal{A}\mb{v}$ when $v_n = -1$, and $\mathcal{S}^{(j)}\mb{v}$ for each $j$ with $(v_j, v_{j+1}) = (-1, 1)$---and the elements covered by $\mb{v}$ are obtained by the inverse operations ($v_n = 1 \to -1$, or some pair $(v_j, v_{j+1}) = (1,-1) \to (-1,1)$); in each case we keep those lying in $Q(n)$.

\smallskip
\noindent\emph{\textup{(II)} Symmetry.}
By \Cref{prop:symQ} negation $\mb{x} \mapsto -\mb{x}$ is an order-reversing bijection of $Q(n)$; being order-reversing, it sends an element of rank $i$ to one of rank $r - i$, where $r = \tfrac{n(n+1)}{2}$ is the maximum rank of $P(n)$. Consequently the elements covered by $-\mb{v}$ are the negatives of those covering $\mb{v}$, and the set of elements incomparable to both $\mb{v}$ and $-\mb{v}$ is unchanged when $\mb{v}$ is replaced by $-\mb{v}$. Thus rows \textup{(8)--(12)}, whose first entries are the negatives of those in rows \textup{(2)--(6)}, follow once the covering elements and the incomparable sets of rows \textup{(2)--(6)} are known; we record the covering elements as we proceed.

\smallskip
\noindent\emph{\textup{(III)} A diamond criterion.}
We use repeatedly the following fact: in a finite poset, if $\mb{y} \prec \mb{x}$, then $\mb{x}$ lies above some element covering $\mb{y}$, since a minimal element of $\{\mb{w} \colon \mb{y} \prec \mb{w} \preceq \mb{x}\}$ covers $\mb{y}$.
Let $-\mb{m}_k = \mb{z}_0, \mb{z}_1, \ldots, \mb{z}_t = \mb{z}$ be such that $\mb{z}_i$ is the unique cover of $\mb{z}_{i-1}$ for $i \in \{1, \ldots, t\}$ and $\rho(\mb{z}) < r/2 - 1$, and let $\mb{z}$ have exactly two covers $\mb{u}, \mb{v}$, where $\mb{u}$ has a unique cover $\mb{b}$ greater than $\mb{v}$. Then the elements of $Q(n)$ incomparable to both $\mb{v}$ and $-\mb{v}$ are exactly $\mb{u}$ and $-\mb{u}$.
The covers $\mb{u}, \mb{v}$ of $\mb{z}$ are distinct, hence incomparable. Moreover $\mb{u}$ and $-\mb{v}$ are incomparable: their ranks $\rho(\mb{u}) = \rho(\mb{v}) < r/2$ and $\rho(-\mb{v}) = r - \rho(\mb{v}) > r/2$ lie on opposite sides of the middle by \textup{(II)}, so the only possible relation is $\mb{u} \preceq_Q -\mb{v}$; but then, $\mb{b}$ being the unique cover of $\mb{u}$, we would have $\mb{v} \preceq_Q \mb{b} \preceq_Q -\mb{v}$, hence $\mb{v} \preceq \mb{0}$, contradicting $\mb{v} \in Q(n)$. By \textup{(II)} the same holds for $-\mb{u}$; thus $\mb{u}$ and $-\mb{u}$ are incomparable to both $\mb{v}$ and $-\mb{v}$.
Conversely, let $\mb{x}$ be incomparable to both $\mb{v}$ and $-\mb{v}$. By \Cref{lem:compara} $\mb{x}$ is comparable to $-\mb{m}_k$ or $\mb{m}_k$, so by \textup{(II)} we may assume $-\mb{m}_k \preceq_Q \mb{x}$; in fact $\mb{z}_0 = -\mb{m}_k \prec_Q \mb{x}$, since $-\mb{m}_k \preceq_Q \mb{z} \prec_Q \mb{v}$ while $\mb{x}$ is incomparable to $\mb{v}$. If $\mb{x} \succ_Q \mb{z}_{i-1}$ for some $i \in \{1, \ldots, t\}$, then $\mb{x}$ lies above some cover of $\mb{z}_{i-1}$, that is, $\mb{x} \succeq_Q \mb{z}_i$; hence $\mb{x}$ is one of $\mb{z}_1, \ldots, \mb{z}_{t-1}$---each less than $\mb{v}$, excluded---or $\mb{x} \succeq_Q \mb{z}$. In the latter case $\mb{x} \neq \mb{z}$ (as $\mb{z} \prec_Q \mb{v}$), so $\mb{x}$ lies above some cover of $\mb{z}$: either $\mb{x} \succeq_Q \mb{v}$, excluded, or $\mb{x} \succeq_Q \mb{u}$, where $\mb{x} \succ_Q \mb{u}$ would force $\mb{x} \succeq_Q \mb{b} \succ_Q \mb{v}$, excluded. Hence $\mb{x} = \mb{u}$.

\smallskip
\noindent\emph{Covers \textup{(all rows)}.}
The covered elements in the second column of \Cref{Tab:compara} follow from rule \textup{(I)}, and \Cref{fig:local} displays the two local structures. Near $-\mb{m}_0$ (\Cref{fig:local}(a)), the elements $-\mb{m}_0$, $\mathcal{A}(-\mb{m}_0)$, and $\mathcal{S}^{(n-1)}\mathcal{A}(-\mb{m}_0)$ form a chain of unique covers (rows \textup{(2)},\textup{(3)}); the last of these has the two covers $\mathcal{S}^{(n-2)}\mathcal{S}^{(n-1)}\mathcal{A}(-\mb{m}_0)$ (row \textup{(4)}) and $\mathcal{A}\mathcal{S}^{(n-1)}\mathcal{A}(-\mb{m}_0)$, and the first of these two is in turn covered by $\mathcal{S}^{(n-3)}\mathcal{S}^{(n-2)}\mathcal{S}^{(n-1)}\mathcal{A}(-\mb{m}_0)$ and $\mathcal{A}\mathcal{S}^{(n-2)}\mathcal{S}^{(n-1)}\mathcal{A}(-\mb{m}_0)$. Near $-\mb{m}_\ell$, the element $-\mb{m}_\ell$ is covered by $\mathcal{S}^{(\ell)}(-\mb{m}_\ell)$ (row \textup{(5)}), and also by $\mathcal{A}(-\mb{m}_\ell)$ when $n$ is even (\Cref{fig:local}(c), row \textup{(6)}); and $\mathcal{S}^{(\ell)}(-\mb{m}_\ell)$ is covered by $\mathcal{S}^{(\ell-1)}\mathcal{S}^{(\ell)}(-\mb{m}_\ell)$ and $\mathcal{S}^{(\ell+1)}\mathcal{S}^{(\ell)}(-\mb{m}_\ell)$, together with $\mathcal{A}\mathcal{S}^{(\ell)}(-\mb{m}_\ell)$ when $n$ is even. By \textup{(II)}, the covered elements in rows \textup{(8)--(12)} are the negatives of the covers just listed: rows \textup{(8)},\textup{(9)},\textup{(10)} negate the covers of $\mathcal{A}(-\mb{m}_0)$, $\mathcal{S}^{(n-1)}\mathcal{A}(-\mb{m}_0)$, and $\mathcal{S}^{(n-2)}\mathcal{S}^{(n-1)}\mathcal{A}(-\mb{m}_0)$, while rows \textup{(11)},\textup{(12)} negate the covers of $\mathcal{S}^{(\ell)}(-\mb{m}_\ell)$.

\smallskip
\noindent\emph{Incomparable sets \textup{(all rows)}.}
\emph{Rows \textup{(1)},\textup{(7)}}: the elements $\pm\mb{m}_k$ are, by \Cref{lem:compara}, comparable to every element, so nothing is incomparable to both.
\emph{Rows \textup{(2)},\textup{(3)}, and Row \textup{(5)} for $n$ odd}: here the row element lies on a unique-cover chain from a minimal element with no branching below it, so every $\mb{x} \succeq_Q -\mb{m}_k$ is comparable to it; with \Cref{lem:compara} and \textup{(II)} the set is empty.
\emph{Row \textup{(4)}}: apply the diamond criterion \textup{(III)} with $\mb{z} := \mathcal{S}^{(n-1)}\mathcal{A}(-\mb{m}_0)$, reached from $-\mb{m}_0$ by a unique-cover chain; its two covers are $\mb{v} := \mathcal{S}^{(n-2)}\mb{z}$ (the row element) and $\mb{u} := \mathcal{A}\mb{z}$, and the unique cover of $\mb{u}$ is $\mathcal{S}^{(n-2)}\mb{u} = \mathcal{A}\mb{v}$, which is greater than $\mb{v}$. Thus \textup{(III)} gives $\{\mb{u}, -\mb{u}\} = \pm\mathcal{A}\mathcal{S}^{(n-1)}\mathcal{A}(-\mb{m}_0)$.
\emph{Row \textup{(6)}} ($n$ even): apply \textup{(III)} with $\mb{z} := -\mb{m}_\ell$ (itself minimal); its two covers are $\mb{v} := \mathcal{S}^{(\ell)}(-\mb{m}_\ell)$ (the row element) and $\mb{u} := \mathcal{A}(-\mb{m}_\ell)$, and the unique cover of $\mb{u}$ is $\mathcal{S}^{(\ell)}\mb{u} = \mathcal{A}\mb{v}$, which is greater than $\mb{v}$. Thus \textup{(III)} gives $\{\mb{u}, -\mb{u}\} = \pm\mathcal{A}(-\mb{m}_\ell)$.
\emph{Rows \textup{(8)--(12)}}: by \textup{(II)} the incomparable sets coincide with those of rows \textup{(2)--(6)}.
All entries of \Cref{Tab:compara} are thereby verified.
\end{proof}

For $n \leq 5$ the poset $Q(n)$ has at most $12$ elements and the neighborhoods of the minimal and maximal elements overlap, so these finitely many cases are immaterial to the polynomial-time conclusion and are omitted.

Combining \Cref{lem:table} with the optimality criterion of \Cref{th:opt} yields polynomially solvable cases for all the tabulated elements, generalizing \Cref{th:min,cor:maximal} to the neighborhoods of the minimal and maximal elements.

\begin{corollary}[Extended polynomially solvable cases]\label{cor:extended}
Let $n \geq 6$ and let $\mb{v}$ be any element listed in the first column of \Cref{Tab:compara}. Denote by $\mb{w}_1, \dots, \mb{w}_s$ the elements covered by $\mb{v}$ and by $\pm\mb{u}_1, \dots, \pm\mb{u}_t$ the elements incomparable to both $\mb{v}$ and $-\mb{v}$, as listed in the second and third columns of the table. If
\begin{romannum}
\item $\mb{v}\cdot \mb{c} \geq 0$,
\item $(\mb{v}+\mb{w}_i)\cdot \mb{c} \leq 0$ for all $i \in \{1, \dots, s\}$, and
\item $(\mb{v}+\mb{u}_j)\cdot \mb{c} \leq 0$ or $(\mb{v}-\mb{u}_j)\cdot \mb{c} \leq 0$ for all $j \in \{1, \dots, t\}$,
\end{romannum}
then the subset $S$ with $\mb{p}(S)=\mb{v}$ is an optimal solution and the optimal value is $\mb{v}\cdot \mb{c}$.
\end{corollary}

Since $s \leq 3$ and $t \leq 1$ for every listed element, this is a condition on a constant number of inequalities; once the instance is sorted, it can be checked in $O(n)$ arithmetic operations, since each inequality is an inner product with a fixed vector. \Cref{th:min,cor:maximal} are the special cases given by rows~(1) and~(7), in which $t=0$ (so condition (iii) is vacuous) and, for row~(1), $s=0$ as well.

\begin{proof}
If $\mb{v}=-\mb{m}_k$ is minimal (row~(1)), then $s=t=0$ and the claim is \Cref{th:min}. Otherwise $\mb{v}$ is non-minimal. By \Cref{lem:table}, the elements covered by $\mb{v}$ are exactly $\mb{w}_1, \dots, \mb{w}_s$ and the elements incomparable to both $\mb{v}$ and $-\mb{v}$ are exactly $\pm\mb{u}_1, \dots, \pm\mb{u}_t$; hence every element of $Q(n)$ other than $\pm\mb{v}$ lies in $N_{\rm C}(\mb{v})$ or equals some $\pm\mb{u}_j$. By conditions (i) and (ii) and \Cref{prop:N}, we have $\mb{v}\cdot\mb{c} \leq \lvert\mb{x}\cdot\mb{c}\rvert$ for every $\mb{x} \in N_{\rm C}(\mb{v})$. For each $j$, condition (iii) gives $\mb{v}\cdot\mb{c} \leq \lvert\mb{u}_j\cdot\mb{c}\rvert = \lvert(-\mb{u}_j)\cdot\mb{c}\rvert$. Therefore $\mb{v}\cdot\mb{c} \leq \lvert\mb{p}(S')\cdot\mb{c}\rvert$ for every $S'$ with $\mb{p}(S')\in Q(n)$, that is, $\mb{v}\cdot\mb{c}$ is the optimal value.
\end{proof}

\begin{example}
Let $n = 6$ and $\mb{c} = (10, 6, 5, 4, 1, 1)$, and consider row (4) of \Cref{Tab:compara}:
\[
\mb{v} = \mathcal{S}^{(4,5)}\mathcal{S}^{(5,6)}\mathcal{A}^{(6)}(-\mb{m}_0) = (1,-1,-1,1,-1,-1),
\]
whose unique covered element is $\mb{w}_1 = (1,-1,-1,-1,1,-1)$ and whose incomparable elements are $\pm\mb{u}_1$ with $\mb{u}_1 = (1,-1,-1,-1,1,1)$. Conditions (i)--(iii) of \Cref{cor:extended} read
$\mb{v}\cdot\mb{c} = 1 \geq 0$;
$(\mb{v}+\mb{w}_1)\cdot\mb{c} = 2(c_1 - c_2 - c_3 - c_6) = -4 \leq 0$;
and $(\mb{v}+\mb{u}_1)\cdot\mb{c} = 2(c_1 - c_2 - c_3) = -2 \leq 0$.
Hence the subset $S$ with $\mb{p}(S) = \mb{v}$, namely $S = \{1, 4\}$, is an optimal solution with optimal value $1$.
\end{example}

\section{Concluding Remarks}\label{sec4}
We have associated two posets with the partition problem: $P(n)$, which is order-isomorphic to the classical poset $M(n)$, and its subposet $Q(n)$. The partial order captures exactly the instance-independent dominance structure of the problem (\Cref{prop:d&order}), the elements of $Q(n)$ are precisely the initial candidate solutions (\Cref{th:sol}), the candidate solutions can be pruned along the partial order (\Cref{th:signQ}), and the optimality criterion of \Cref{th:opt} yields an explicitly described family of polynomially solvable cases (\Cref{cor:extended}).

In retrospect our results are simple; the difficulty lay in finding the right formulation, and two choices turned out to be essential.
The first is to encode subsets by vectors with entries in $\{1, -1\}$ rather than $\{1, 0\}$: with the latter encoding the complementation $S \mapsto [n]\setminus S$ does not correspond to the negation $\mb{v} \mapsto -\mb{v}$, and the symmetry of $Q(n)$ (\Cref{prop:symQ}), which we use throughout, is lost.
The second is to refrain from halving the candidate solutions by normalizations such as requiring every subset to contain the element $1$: such a normalization destroys the sign-opposite pairs $\{\mb{v}, -\mb{v}\}$ on which the pruning rules of \Cref{th:signQ} and the criterion of \Cref{th:opt} rest.

We also record a connection that places the poset $Q(n)$ in a classical context: threshold logic and simple games.
For an instance $\mb{c}$ with $c_1 > \cdots > c_n > 0$ such that $\mb{v}\cdot \mb{c} \neq 0$ for every $\mb{v} \in Q(n)$, the set $F_{\mb{c}} = \{\mb{v} \in Q(n) \colon \mb{v}\cdot \mb{c} > 0\}$ is an up-set of $Q(n)$ by \Cref{prop:d&order}, and it is \emph{antisymmetric}: exactly one of $\mb{v}$ and $-\mb{v}$ belongs to $F_{\mb{c}}$.
The antisymmetric up-sets of $Q(n)$ are in one-to-one correspondence with the self-dual Boolean functions of $n$ variables that are \emph{regular} with respect to the variable order $x_1 \succeq \cdots \succeq x_n$---meaning that the family of subsets on which the function equals $1$ is an up-set of $P(n)$ under the dominance order---\cite{Crama2011,Muroga1971}. Indeed, identify a Boolean function with the family of subsets of $[n]$ on which it takes the value $1$, and each subset $S$ with $\mb{p}(S) \in P(n)$: self-duality states that exactly one of $\mb{v}, -\mb{v}$ is in the family, and regularity states that the family is an up-set of $P(n)$, since $\preceq_P$ coincides with the dominance order $\#(S \cap [k]) \leq \#(T \cap [k])$ for all $k$. An antisymmetric up-set $F$ of $P(n)$ must contain $R_+(n)$ and avoid $R_-(n)$---if $\mb{v} \in R_+(n)$ and $-\mb{v} \in F$, then $-\mb{v} \preceq_P \mb{v}$ puts $\mb{v}$ in $F$ as well, contradicting antisymmetry---so $F$ is exactly $R_+(n) \cup G$ for an antisymmetric up-set $G$ of $Q(n)$, and every such union is an antisymmetric up-set of $P(n)$.
Under this correspondence the up-sets of the form $F_{\mb{c}}$, that is, those realizable by an instance, correspond to the self-dual \emph{threshold} functions, equivalently to the weighted majority games \cite{Taylor1999,Neumann1944}: the subset $S$ wins the game with weights $c_i$ and quota $\frac12\sum_i c_i$ exactly when $\mb{p}(S)\cdot \mb{c} > 0$.
Not every antisymmetric up-set is realizable, however: a classical result of threshold logic states that every regular function of at most eight variables is a threshold function, while from nine variables on this fails \cite{Crama2011,Muroga1971}. Accordingly, the number of antisymmetric up-sets of $Q(n)$ (sequence A109456 in \cite{OEIS}) and the number of realizable ones (sequence A001532 in \cite{OEIS}, the number of weighted majority games) agree for $n \leq 8$, where both count $2, 3, 7, 21, 135, 2470$ for $n = 3, \dots, 8$, and diverge at $n = 9$. In partition-problem terms, since the realizable up-sets are exactly the threshold ones, from $n = 9$ on there exist antisymmetric up-sets of $Q(n)$---globally consistent sign patterns---that no instance $\mb{c}$ realizes. A systematic study of this correspondence---which up-sets the instances of the partition problem realize, and what the realizing instances reveal about the problem---is left for future work.

We conclude with two open problems.
\begin{enumerate}
\item How far can the explicit family of \Cref{lem:table} be extended? More precisely, characterize the elements $\mb{v} \in Q(n)$ for which the number of elements incomparable to both $\mb{v}$ and $-\mb{v}$ is bounded by a constant (or by a polynomial in $n$), so that the criterion of \Cref{th:opt} can be evaluated in polynomial time.
\item Determine the growth order of $\max_{\mb{v} \in Q(n)} \gamma(\mb{v})$, the largest number of candidate solutions guaranteed to be discarded by a single sign test. By \Cref{prop:singlequery} it is $O(2^n/n^{1/4})$, and the values in \Cref{tab:min} are consistent with an order $2^n/n^{\theta}$ for a constant $\theta$; if such an order exists, then $\theta \geq 1/4$ by \Cref{prop:singlequery}, and the exact exponent is unknown to us.
\end{enumerate}

\section*{Acknowledgments}
The author is grateful to Professor Kazuhisa Makino (Kyoto University) for helpful comments and suggestions.

\section*{Declarations}
During the preparation of this work the author used Claude (Anthropic)
to refine the English exposition and to check the clarity of the
presentation. After using this tool, the author reviewed and edited the
content as needed. The author assumes responsibility for all content.

\bibliographystyle{siamplain}
\bibliography{reference}
\end{document}